\documentclass[a4paper,fleqn]{cas-sc}

\usepackage[numbers]{natbib}

\usepackage{amssymb}
\usepackage{amsmath,amsthm}
\usepackage[noend]{algpseudocode}
\usepackage{algorithm} 
\usepackage{array}
\usepackage[caption=false,font=footnotesize]{subfig}
\usepackage{stfloats}
\usepackage{url}
\usepackage{cases}
\usepackage{stfloats}
\usepackage{upgreek}
\usepackage{balance}
\usepackage{color}
\usepackage[dvipsnames]{xcolor}
\usepackage{verbatim}

\usepackage{graphicx}

\usepackage[T1]{fontenc} 

\usepackage{threeparttable}
\usepackage{multirow}

\usepackage{siunitx}
\usepackage{xcolor}
\usepackage[normalem]{ulem}
\usepackage{mleftright}
\usepackage[]{footmisc}
\mleftright
\newcommand{\sizecorr}[1]{\makebox[0cm]{\phantom{$\displaystyle #1$}}}
\DeclareMathOperator{\vect}{vec}
\DeclareMathOperator{\trace}{Tr}

\newcommand{\paren}[1]{\left({#1}\right)}
\newcommand{\bracket}[1]{{\left [{#1}\right ]}}
\newcommand{\braces}[1]{{\left\{ {#1}\right\}}} 
\newcommand{\ith}[1]    {{#1}{\text{-th}}}
\newcommand{\rr}{_\mathrm{r}}
\newcommand{\cc}{_\mathrm{c}}

\newcommand{\B}{\textrm{B}}
\newcommand{\rnr}{_{\mathrm{r},n_\mathrm{r}}}
\newcommand{\target}{\mathrm{t}}

\newcommand{\Ernr}{\mathbf{E}_{\textrm{r},n_{\textrm{r}}}}

\newcommand{\EiB}{\mathbf{E}_{\textrm{u},i}\bracket{k}}
\newcommand{\EiBn}{\mathbf{E}_{\textrm{u},i}\bracket{k}}
\newcommand{\EiBop}{\mathbf{E}^\star_{\textrm{u},i}\bracket{k}}
\newcommand{\EBj}{\mathbf{E}_{\textrm{d},j}\bracket{k}}
\newcommand{\EBjone}{\mathbf{E}_{\textrm{d},j}\bracket{k}}
\newcommand{\EBjop}{\mathbf{E}^\star_{\textrm{d},j}\bracket{k}}

\newcommand{\duis}{\mathbf{d}_{\mathrm{u},i}\bracket{k}}

\newcommand{\ddjs}{\mathbf{d}_{\mathrm{d},j}\bracket{k}}

\newcommand{\yui}{\mathbf{y}_{\textrm{u},i}\bracket{k}}

\newcommand{\ydj}{\mathbf{y}_{\textrm{d},j}\bracket{k}}

\newcommand{\PiB}{\mathbf{P}_{\textrm{u},i}\bracket{k}}
\newcommand{\PiBH}{\mathbf{P}^\dagger_{\textrm{u},i}\bracket{k}}
\newcommand{\PqB}{\mathbf{P}_{\textrm{u},q}\bracket{k}}
\newcommand{\PqBH}{\mathbf{P}^\dagger_{\textrm{u},q}\bracket{k}}
\newcommand{\PBj}{\mathbf{P}_{\textrm{d},j}\bracket{k}}

\newcommand{\PBjH}{\mathbf{P}^\dagger_{\textrm{d},j}\bracket{k}}
\newcommand{\PBg}{\mathbf{P}_{\textrm{d},g}\bracket{k}}
\newcommand{\PBgH}{\mathbf{P}^\dagger_{\textrm{d},g}\bracket{k}}

\newcommand{\Rj}{\mathbf{R}_{\textrm{d},j}\bracket{k}}

\newcommand{\Rinjin}{\left( \mathbf{R}^{\textrm{in}}_{\mathrm{d},j}\bracket{k}\right)^{-1}}

\newcommand{\Ringin}{\left( \mathbf{R}^{\textrm{in}}_{\mathrm{d},g}\bracket{k}\right)^{-1}}
\newcommand{\Rjs}{\mathbf{R}_{\mathrm{d},j}\bracket{k}}

\newcommand{\Rinjins}{\left( \mathbf{R}^{\mathrm{in}}_{\mathrm{d},j}\bracket{k}\right)^{-1}}

\newcommand{\Ri}{\mathbf{R}_{\textrm{u},i}\bracket{k}}

\newcommand{\Riniin}{\left( \mathbf{R}^\mathrm{in}_{\mathrm{u},i}\bracket{k}\right)^{-1}}

\newcommand{\Rinqin}{\left( \mathbf{R}^{\textrm{in}}_{\mathrm{u},q}\bracket{k}\right)^{-1}}
\newcommand{\Ris}{\mathbf{R}_{\mathrm{u},i}\bracket{k}}

\newcommand{\UqB}{\mathbf{U}_{\textrm{u},q}\bracket{k}}

\newcommand{\UiB}{\mathbf{U}_{\textrm{u},i}\bracket{k}}

\newcommand{\UiBH}{\mathbf{U}^\dagger_{\textrm{u},i}\bracket{k}}

\newcommand{\UqBnH}{\mathbf{U}^\dagger_{\textrm{u},q}\bracket{k}}
\newcommand{\WiB}{\mathbf{W}_{\textrm{u},i}\bracket{k}}
\newcommand{\WiBn}{\mathbf{W}_{\textrm{u},i}\bracket{k}}
\newcommand{\WqB}{\mathbf{W}_{\textrm{u},q}\bracket{k}}

\newcommand{\UBj}{\mathbf{U}_{\textrm{d},j}\bracket{k}}

\newcommand{\UBjH}{\mathbf{U}^\dagger_{\textrm{d},j}\bracket{k}}

\newcommand{\WBj}{\mathbf{W}_{\mathrm{d},j}\bracket{k}}
\newcommand{\WBjop}{\mathbf{W}^\star_{\mathrm{d},j}\bracket{k}}
\newcommand{\WBjone}{\mathbf{W}_{\textrm{d},j}\bracket{k}}

\newcommand{\Wrnr}{\mathbf{W}_{\mathrm{r},n_\mathrm{r}}}
\newcommand{\urk}{\mathbf{u}_{\mathrm{r},n_\mathrm{r}}\bracket{k}}

\newcommand{\HrB}{\mathbf{H}_{\textrm{rB}}}
\newcommand{\HrBH}{\mathbf{H}^\dagger_{\textrm{rB}}}
\newcommand{\Hrj}{\mathbf{H}_{\textrm{r},j}}
\newcommand{\Hrg}{\mathbf{H}_{\textrm{r},g}}
\newcommand{\HrjH}{\mathbf{H}^\dagger_{\textrm{r},j}}
\newcommand{\HrgH}{\mathbf{H}^\dagger_{\textrm{r},g}}
\newcommand{\HBj}{\mathbf{H}_{\textrm{B},j}}
\newcommand{\HBjH}{\mathbf{H}^\dagger_{\textrm{B},j}}
\newcommand{\HBg}{\mathbf{H}_{\textrm{B},g}}
\newcommand{\HBgH}{\mathbf{H}^\dagger_{\textrm{B},g}}
\newcommand{\HBB}{\mathbf{H}_{\mathrm{BB}}}
\newcommand{\HBBH}{\mathbf{H}^\dagger_{\mathrm{BB}}}
\newcommand{\HiB}{\mathbf{H}_{i,\textrm{B}}}
\newcommand{\HiBH}{\mathbf{H}^\dagger_{i,\textrm{B}}}
\newcommand{\HqB}{\mathbf{H}_{q,\textrm{B}}}
\newcommand{\HqBH}{\mathbf{H}^\dagger_{q,\textrm{B}}}
\newcommand{\Hij}{\mathbf{H}_{i,j}}
\newcommand{\HijH}{\mathbf{H}^\dagger_{i,j}}

\newcommand{\sfrac}[2]{#1/#2}

\newtheorem{theorem}{Theorem}

\newtheorem{thm}{Theorem}

\newtheorem{prop}[thm]{Proposition}

\theoremstyle{definition}

\begin{document}

\shorttitle{Distributed MRMC - II}


\title [mode = title]{Co-Designing Statistical MIMO Radar and In-band Full-Duplex Multi-User MIMO Communications -- Part II: Joint Precoder, Radar Code, and Receive Filters Design}    

\tnotetext[1]{K. V. M. acknowledges support from the National Academies of Sciences, Engineering, and Medicine via the Army Research Laboratory Harry Diamond Distinguished Fellowship. The following co-authors are a member of EURASIP: Kumar Vijay Mishra.}

\author[label1]{Jiawei~Liu}
\author[label2]{Kumar Vijay Mishra}[orcid=0000-0002-5386-609X]
\author[label1]{Mohammad~Saquib}

\affiliation[label1]{organization={The University of Texas at Dallas},
            city={Richardson},
            postcode={TX 75080}, 
            country={USA}}
\affiliation[label2]{organization={United States DEVCOM Army Research Laboratory},
            city={Adelphi},
            postcode={MD 20783}, 
            country={USA}}

	\maketitle
\begin{abstract}
We address the challenge of spectral sharing between a statistical multiple-input multiple-output (MIMO) radar and an in-band full-duplex (IBFD) multi-user MIMO (MU-MIMO) communications system operating simultaneously in the same frequency band. Existing research on joint MIMO-radar-MIMO-communications (MRMC) systems has limitations, such as focusing on colocated MIMO radars, half-duplex MIMO communications, single-user scenarios, neglecting practical constraints, or employing separate transmit/receive units for MRMC coexistence. This paper, along with companion papers (Part I and III), proposes a comprehensive MRMC framework that addresses all these challenges. In the previous companion paper (Part I), we presented signal processing techniques for a distributed IBFD MRMC system. In this paper, we introduce joint design of statistical MIMO radar codes, uplink/downlink precoders, and corresponding receive filters using a novel metric called compounded-and-weighted sum mutual information. To solve the resulting highly non-convex problem, we employ a combination of block coordinate descent (BCD) and alternating projection methods. Numerical experiments show convergence of our algorithm, mitigation of uplink interference, and stable data rates under varying noise levels, channel estimate imperfections, and self-interference. The subsequent companion paper (Part III) extends the discussion to multiple targets and evaluates the tracking performance of our MRMC system.
\end{abstract}


\section{Introduction}
The increasing congestion of the electromagnetic spectrum in recent years has presented significant challenges in the design of radar and communications systems operating within the same frequency bands \cite{mishra2019toward}. While radar systems necessitate substantial transmit signal bandwidths for high-resolution target detection \cite{skolnik2008radar}, wireless cellular networks require access to a broad spectrum to support high data rates \cite{Multiuser,cover2006elements}. In response to the exponential growth of mobile data traffic, network operators globally have turned to higher frequency spectra to accommodate the surge in data usage [1]. Additionally, advancements in wireless communications and the continuous increase in carrier frequencies have prompted spectrum regulators such the Federal Communications Commission (FCC) and International Telecommunications Union (ITU) to grant civilian communications systems access to frequency bands traditionally reserved for radar and sensing applications. This policy shift has initiated a trend of coexistence and convergence between radar and communications functions \cite{mishra2019toward}.

The literature \cite{mishra2019toward} suggests two possible approaches toward joint radar-communications. In the coexistence approach, the radar and communications systems operate as separate entities within the same spectrum using different waveforms \cite{interferencealignment,ayyar2019robust}. In the co-design paradigm, the two systems are integrated into a single hardware platform, and a common waveform is employed at either the transmitter (Tx), receiver (Rx), or both \cite{dokhanchi2019mmwave,duggal2020doppler}. The effectiveness of these spectrum-sharing solutions depends on the level of cooperation between the radar and communications systems. A hybrid approach of \textit{spectral cooperation} has also been suggested, wherein some information exchange may take place between radar and communications systems \cite{MCMIMO_RadComm,he2019performance}. In this paper, we focus on spectral co-design aspects.

The aforementioned approaches do not readily extend to multiple-input multiple-output (MIMO) configurations, which employ multiple antennas at the transmitter and receiver to achieve high spectrum efficiency \cite{Multiuser,haimovich2008mimo}. MIMO configurations enhance communications capacity, provide spatial diversity, and exploit multipath propagation \cite{Multiuser}. Similarly, MIMO radars offer advantages over equivalent phased array radars, such as higher angular resolution with fewer antennas, spatial diversity, and improved parameter identification by leveraging waveform diversity \cite{fisher2006MIMO}. In a colocated MIMO radar \cite{li2007mimo}, the radar cross-section remains the same for closely-spaced antennas. In contrast, in a widely distributed or statistical MIMO radar, the antennas are sufficiently separated from each other, causing the same target to exhibit different radar cross-sections to each Tx-Rx pair. This spatial diversity is advantageous for detecting targets with small backscatters and low speed \cite{sun2024widely}. 

The increased degrees-of-freedom (DoFs), aperture sharing, and higher-dimensional optimization further complicate spectrum sharing in a joint MIMO-radar-MIMO-communications (MRMC) system \cite{alaee2020information,dokhanchi2020multi}. MRMC processing techniques include employing orthogonal transmit waveforms \cite{bao2019precoding} and receiver interference cancellation \cite{khawar2015target}; see, e.g.,  \cite{liu2024codesigningpart1} for a survey on MRMC solutions. Prior MRMC literature primarily focused on single-user MIMO communications and colocated MIMO radars. Co-design with statistical MIMO radar remains relatively unexamined in these prior works. In the previous companion paper (Part I) \cite{liu2024codesigningpart1}, we proposed spectral co-design of statistical MIMO radar with in-band full-duplex (IBFD) multi-user (MU) MIMO communications. The IBFD technology has been recently explored for joint radar-communications systems to facilitate communications transmission while also receiving the target echoes \cite{Barneto2021FDcommsensing}. 

The performance metrics to design radar and communications systems are not identical because of different system goals \cite{mishra2019toward}. As a result, recent works \cite{alaee2020information,dokhanchi2020multi}  have suggested mutual information (MI) as a common metric for joint radar and communications systems. Our previous companion paper (Part I) \cite{liu2024codesigningpart1} proposed MI-inspired novel \textit{compounded-and-weighted sum MI} (CWSM) for the MIMO radar and IBFD MU-MIMO communications co-design problem. In that paper, we described the receive signal processing for a co-designed distributed MRMC system but did not develop an algorithm to solve the design problem. In this paper, unlike many prior works that focus solely on one specific system goal and often in isolation with other processing modules, we propose using CWSM to jointly design the UL/DL precoders, MIMO radar waveform matrix, and linear receive filters (LRFs) for both systems.  Our co-design also accounts for several practical constraints, including the maximum UL/DL transmit powers, the QoS of the UL/DL quantified by their respective minimum achievable rates, and the peak-to-average-power-ratio (PAR) of the MIMO radar waveform. It is common among communications literature to identify a UL/DL UE's QoS with its minimum achievable rate \cite{MIMOCOMSecrecy,biswas2018fdqos}. Adopting low PAR waveforms is crucial for achieving energy- and cost-efficient RF front-ends \cite{NaghshTSP2017}. We address the non-convex CWSM maximization problem's challenges subject to non-convex constraints, namely the QoS and the PAR constraints, by developing an alternating algorithm that incorporates both the block coordinate descent (BCD) and the alternating projection (AP) methods. The BCD-AP process breaks the original problem into less complex subproblems that we iteratively solve. Numerical experiments show a quick, monotonic convergence of our proposed algorithm. Preliminary results of this work appeared in our conference publication \cite{liu2022transceiver}, where only communications design was considered, PAR constraint was excluded, and detailed theoretical guarantees were excluded. 

The rest of the paper is organized as follows. In the next section, we summarize the system model following the details included in the previous companion paper (Part I) \cite{liu2024codesigningpart1}. Then, we formulate the CWSM maximization problem in Section \ref{sec: formulation}. 
We develop the BCD-AP MRMC procedure to solve the non-convex optimization optimization problem iteratively in Section \ref{sec:solution}. We validate the proposed technique through numerical experiments in Section \ref{sec:numerical} before concluding in Section \ref{sec:conclusion}.
	
Throughout this paper, lowercase regular, lowercase boldface and uppercase boldface letters denote scalars, vectors and matrices, respectively. We use $I(\mathbf{X};\mathbf{Y})$ and $H\paren{\mathbf{X}|\mathbf{Y}}$ to denote, MI and conditional entropy between two random variables $\mathbf{X}$ and $\mathbf{Y}$, respectively. The notations $\mathbf{Y}\bracket{k}$, $\mathbf{y}\bracket{k}$, and $y\bracket{k}$ denote the value of time-variant matrix $\mathbf{Y}$, vector $\mathbf{y}$ and scalar $y$ at discrete-time index $k$, respectively; $\mathbf{1}_{N}$ is a vector of size $N$ with all ones; $\mathbb{C}$ and $\mathbb{R}$ represent sets of complex and real numbers, respectively; a circularly symmetric complex Gaussian (CSCG) vector $\mathbf{q}$ with $N$ elements and power spectral density $\mathcal{N}_0$ is $\mathbf{q}\sim\mathcal{CN}(0,\mathcal{N}_0\mathbf{I}_{N})$; $(\cdot)^{\star}$ is the solution of the optimization problem; $\mathbb{E}\bracket{\cdot}$ is the statistical expectation; $\trace\{\mathbf{R}\}$, $\mathbf{R}^\top$, $\mathbf{R}^\dagger$, $\mathbf{R}^\ast$, $\left| \mathbf{R}\right|$, $\mathbf{R}\succeq\mathbf{0}$, and $\mathbf{R}\paren{m,n}$ are the trace, transpose, Hermitian transpose, element-wise complex conjugate, determinant, positive semi-definiteness and $\ith{\paren{m,n}}$ entry of matrix $\mathbf{R}$, respectively; set $\mathbb{Z}_{+}(L)$ denotes $\left\lbrace1,\dots,L\right\rbrace$;  $\mathbf{x}\succeq\mathbf{y}$ denotes component-wise inequality between vectors $\mathbf{x}$ and $\mathbf{y}$; $x^+$ represents $\max(x,0)$; $x^{\paren{t}}\paren{\cdot}$ is the $\ith{t}$ iterate of an iterative function $x\paren{\cdot}$; $\inf(\cdot)$ is the infimum of its argument; $\odot$ denotes the Hadamard product; and $\oplus$ is the direct sum. 

\section{Spectral Co-Design System Model}
\label{sec:system}
The system model for the statistical MIMO The signal model used in this paper closely follows that detailed in the previous companion paper (Part I) \cite{liu2024codesigningpart1} and, hence, we only summarize the key aspects below.

Consider a two-dimensional (2-D) $\left(x\textrm{-}y \right)$ Cartesian plane on which the $M_\mathrm{r}$ Txs and $N_\mathrm{r}$ Rxs of a statistical MIMO radar, the BS, $I$ UL UEs, and $J$ DL UEs of the IBFD MU-MIMO communications system are located at the coordinates $\left(x_{m_\mathrm{r}},y_{m_\mathrm{r}}\right)$, $\left(x_{n_\mathrm{r}},y_{n_\mathrm{r}} \right)$,  $\paren{x_{\mathrm{B}},y_{\mathrm{B}}}$, $\paren{x_{\textrm{UL},i},y_{\textrm{UL},i}}$, and $\paren{x_{\mathrm{DL,}j},y_{\textrm{DL},j}}$, respectively, for all $m_\mathrm{r}\in{Z}_{+}(M_\mathrm{r})$, $n_\mathrm{r}\in\mathbb{Z}_{+}\paren{N_\mathrm{r}}$, $i\in\mathbb{Z}_{+}\paren{I}$, and $j\in\mathbb{Z}_{+}\paren{J}$. The statistical MIMO radar operates within the same transmit spectrum as an IBFD MU-MIMO communications system. Here, the radar aims to detect a target moving within the cellular coverage of the BS. The communications system serves the UL/DL UEs with desired achievable rates in the presence of the radar echoes. 
\subsection{Transmit Signal}
Each radar Tx emits a train of $\mathit{K}$ pulses at a uniform pulse repetition interval (PRI) $T_{\mathrm{r}}$ or \textit{fast-time}; the total duration $KT_{\mathrm{r}}$ is the \textit{coherent processing interval} (CPI) or \textit{slow-time} and $\mathit{K}$ is chosen to avoid range migration during the CPI \cite{skolnik2008radar}. At the same time,  the BS and each UL UE continuously transmit DL and UL symbols, respectively. The radar pulse width is $T_\mathrm{p}= T_\mathrm{r}/N$, where $N$ is the number of sampled range bins or cells in a PRI. The UL/DL frame duration $T_\mathrm{f}$ and the UL/DL symbol duration $T_{\mathrm{s}}$ equal radar PRI and radar pulse width $T_{\mathrm{p}}$, respectively; i.e., $T_{\mathrm{f}}=T_{\mathrm{r}}$ and $T_{\mathrm{s}}=T_{\mathrm{p}}$. This implies that the number of UL/DL frames transmitted in the scheduling window is also $\mathit{K}$ and the number of UL/DL symbols per frame is $\sfrac{T_{\mathrm{f}}}{T_{\mathrm{s}}}=N$. The $\ith{k}$ communications frame is transmitted at a duration of $GT_{\mathrm{p}}$, $G\in\mathbb{Z}\paren{N-1}$, before the $\ith{k}$ radar PRI. 

\color{black}
\subsubsection{Statistical MIMO Radar}
\label{sec: Radar_Tx_Signals}
Denote the \textit{narrowband} transmit pulse of the $m\rr$-th radar Tx by $\phi_{m_\mathrm{r}}\paren{t}$.  
The waveforms from all Txs form the waveform vector
\begin{align}
\boldsymbol{\phi}(t)=\left[ \phi_1(t),\dots,\phi_{M_\mathrm{r}}(t)\right]^\top\in\mathbb{C}^{M\rr},
\end{align}
and satisfy the orthonormality $\int_{T_\mathrm{p}}^{}\boldsymbol{\phi}(t)\boldsymbol{\phi}^\dagger(t)dt=\mathbf{I}_{M_\mathrm{r}}$. The radar code to modulate the pulse emitted by the $m\rr$ Tx in the $\ith{k}$ PRI is $a_{m\rr,k}$. During the observation window $t\in\bracket{0,KT\rr+GT_{\mathrm{p}}}$, the $\ith{m_\mathrm{r}}$ Tx emits the pulse train
\begin{align}
s_{m_\mathrm{r}}\paren{t}=\sum_{k=0}^{K-1}a_{m_\mathrm{r},k}\phi_{m_\mathrm{r}}\paren{t-kT\rr-GT_{\mathrm{p}}},\label{eq: radar_pulse_train}
\end{align}
where the support of $\phi_{m\rr}\paren{t}$ is $\left[0,T\rr\right)$ and, without loss of generality, $\phi_{m\rr}\paren{t}=\sqrt{\sfrac{1}{T_{\textrm{p}}}}e^{j2\pi\frac{m\rr}{T_{\textrm{p}}}t}$ for $m\rr\in\mathbb{Z}_+\paren{M\rr}$ for $t\in\left[0,T\rr\right)$. Define the radar code vector transmitted during the $\ith{k}$ PRI as $\mathbf{a}\bracket{k}=\bracket{a_{1,k},\cdots,a_{\mathit{M}\rr,k}}^\top\in\mathbb{C}^{M\rr}$ so that the MIMO radar code matrix is
\begin{align}
\mathbf{A}=\bracket{\mathbf{a}^\top\bracket{1};\cdots; \mathbf{a}^\top\bracket{\mathrm{\mathit{K}}}}=\bracket{\mathbf{a}_1,\cdots.\mathbf{a}_{M\rr}}\in\mathbb{C}^{\mathit{K}\times \mathit{M}\rr}. 
\end{align}
where $\mathbf{a}_{m\rr}\in\mathbb{C}^{K}$ is the code of the $\ith{m\rr}$ TX over all PRIs. The combined transmit signal vector is
\begin{align}
\mathbf{s}(t)=\bracket{s_1(t),\cdots,s_{M_\mathrm{r}}(t)}^\top\in\mathbb{C}^{M\rr}. 
\end{align}

\subsubsection{IBFD MU-MIMO Communications}
The BS and UEs operate in the FD and HD modes, respectively. During the observation window, the BS receives data frames from the $I$ UL UEs; concurrently, the $J$ DL UEs operating in the same band download data frames from the BS. The BS is equipped with $\mathit{M}_\mathrm{c}$ transmit and $N_{\mathrm{c}}$ receive antennas. The $i$-th UL UE and $j$-th DL UE  employ $N^{\textrm{u}}_{i}$ and $N^{\textrm{d}}_{j}$ transceive antennas, respectively. To achieve the maximum capacities of the UL and DL channels, number of BS Tx and Rx antennas are $\mathit{M}\cc\geq\sum_{j=1}^{\mathit{J}}N^{\textrm{d}}_{j}$ and $N\cc\geq\sum_{i=1}^{\mathit{I}}N^{\textrm{u}}_{i}$, respectively \cite{Multiuser}. 
A total of $\mathit{D}^{\textrm{u}}_{i}\leq N^{\textrm{u}}_{i}$ and $\mathit{D}^{\textrm{d}}_{j}\leq N^{\textrm{d}}_{j}$ unit-energy data streams are used by $i$-th UL UE and $j$-th DL UE, respectively. The symbol vectors sent by the $i$-th UL UE toward the BS and by the BS toward the $j$-th DL UE in the $\ith{l}$ symbol period of the $\ith{k}$ frame are $\mathbf{d}_{\mathrm{u},i}\bracket{k,l}\in \mathbb{C}^{D^\textrm{u}_{i}}$ and $\mathbf{d}_{\mathrm{d},j}\bracket{k,l}\in \mathbb{C}^{\mathit{D}^{\textrm{d}}_{j}}$, respectively; these are independent and identically distributed (i.i.d.) with $\mathbb{E}\bracket{\mathbf{d}_{\mathrm{d},j}\mathbf{d}^\dagger_{\mathrm{d},j}\bracket{k,l}}=\mathbb{E}\bracket{\mathbf{d}_{\mathrm{u},i}\mathbf{d}^\dagger_{\mathrm{u},i}\bracket{k,l}}=\mathbf{I}$ for $i\in\mathbb{Z}_+\braces{\mathit{I}}$, $k\in\mathbb{Z}_+\braces{\mathit{K}}$, and $l\in\mathbb{Z}_+\braces{N}$.

Denote the precoders for the $i$-th UL UE and the $j$-th DL UE at the $\ith{k}$ frame as $\PiB\in\mathbb{C}^{N^{\textrm{u}}_{i}\times \mathit{D}^{\textrm{u}}_{i}}$ and $\PBj\in\mathbb{C}^{\mathit{M}\cc\times \mathit{D}^{\textrm{d}}_{j}}$, respectively. The precoded transmit signal vectors for the $i$-th UL UE and $j$-th DL UE 
become 
\begin{align}
\mathbf{s}_{\textrm{u},i}\bracket{k,l}=\PiB\mathbf{d}_{\mathrm{u},i}\bracket{k,l}, 
\end{align}
and 
\begin{align}
\mathbf{s}_{\textrm{d},j}\bracket{k,l}=\PBj\mathbf{d}_{\mathrm{d},j}\bracket{k,l},
\end{align}
respectively. The total DL symbol vector broadcast by the BS in the same symbol period is 
\begin{align}
\mathbf{s}_{\mathrm{B}}\bracket{k,l}=\sum_{j=1}^{J}\mathbf{s}_{\textrm{d},j}\bracket{k,l}.
\end{align}
The transmit pulse shaping function used by the IBFD communications is $p_{\mathrm{T}}\paren{t}$. The transmit signals of $i$-th UL UE and BS are
\begin{flalign}
\mathbf{x}_{\mathrm{u},i}\paren{t}&=\sum_{k=0}^{K-1}\sum_{l=0}^{N-1}\mathbf{s}_{\mathrm{u},i}\bracket{k,l}p_{\mathrm{T}}\paren{t-(kN+l)T_{\mathrm{p}}}, \end{flalign}
and
\begin{flalign}
\mathbf{x}_{\mathrm{B}}\paren{t}&=\sum_{k=0}^{K-1}\sum_{l=0}^{N-1}\mathbf{s}_{\mathrm{B}}\bracket{k,l}p_{\mathrm{T}}\paren{t-\paren{k N+l}T_{\mathrm{p}}}.
\end{flalign}

\subsection{Statistical MIMO radar receiver}
\label{subsec:radar_total_receive}
The radar and DL signals are reflected off a single target and the combined echo is received at the $\ith{n\rr}$ radar Rx as $\mathbf{y}_{\mathrm{t},n\rr}$. It is overlaid with clutter echoes $\mathbf{y}_{\mathrm{c},n\rr}$, directly received IBFD MU-MIMO DL signal $\mathbf{y}_{\mathrm{Bm},n\rr}$, and interference from the UL $\mathbf{y}_{\mathrm{u},n\rr}$. With the CSCG noise vector at the $\ith{n\rr}$ radar Rx by $\mathbf{z}\rnr\in\mathcal{CN}\paren{\mathbf{0},\sigma^2\rnr\mathbf{I}_{K}}$, the composite receive signal model at the range cell under test (CUT) of the $\ith{n\rr}$ radar Rx is
\begin{flalign}
\mathbf{y}\rnr=\mathbf{y}_{\mathrm{t},n\rr}+\underbrace{\mathbf{y}_{\mathrm{c},n\rr}+\mathbf{y}_{\mathrm{Bm},n\rr}+\mathbf{y}_{\mathrm{u},n\rr}+\mathbf{z}\rnr}_{=\mathbf{y}^{\mathrm{in}}_{\mathrm{r},n\rr}},\label{eq:combined_rad_rx}
\end{flalign}\normalsize
where $\mathbf{y}^{\mathrm{in}}_{\mathrm{r},n\rr}$ denotes the interference-plus-noise component of $\mathbf{y}\rnr$. The covariance matrix (CM) of $\mathbf{y}\rnr$ is \cite{liu2024codesigningpart1}
\begin{align}
\mathbf{R}_{\textrm{r},n\rr}=\mathbf{R}_{\target,n\rr}+\mathbf{R}^{\mathrm{in}}_{\mathrm{r},n\rr},
\end{align}
with the CM of $\mathbf{y}^{\mathrm{in}}_{\textrm{r},n\rr}$ given by
\begin{align}
\mathbf{R}^{\mathrm{in}}_{\mathrm{r},n\rr}\triangleq\mathbf{R}_{\textrm{c},n\rr}+\mathbf{R}_{\mathrm{Bm},n\rr}+\mathbf{R}_{\mathrm{Ur},n\rr}+\sigma^2\rnr\mathbf{I}_{K}.
\end{align}

Combining the received signals from $N\rr$ radar Rxs yields
\begin{align}
\mathbf{y}_{\mathrm{r}}=\mathbf{y}_{\textrm{tr}}+\mathbf{y}^{\textrm{in}}_{\textrm{r}}=\bracket{\mathbf{y}^\top_{\textrm{r},1};\cdots;\mathbf{y}^\top_{\textrm{r},N\rr}}^\top\in\mathbb{C}^{KN\rr},
\end{align}
where 
\begin{align}
\mathbf{y}_{\textrm{tr}}=\bracket{\mathbf{y}^\top_{\textrm{t},1};\cdots;\mathbf{y}^\top_{\textrm{t},N\rr}}^\top,
\end{align}
and 
\begin{align}
\mathbf{y}^{\textrm{in}}_{\textrm{r}}\triangleq\mathbf{y}_{\textrm{cr}}+\mathbf{y}_{\textrm{Bmr}}+\mathbf{y}_{\textrm{Ur}}+\mathbf{z}_{\textrm{r}}=\bracket{\paren{\mathbf{y}^{\textrm{in}}_{\textrm{r},1}}^\top;\cdots;\paren{\mathbf{y}^{\textrm{in}}_{\textrm{r},N\rr}}^\top}^\top,
\end{align}
whose CM is 
\begin{align}
\mathbf{R}^{\textrm{in}}_{\textrm{r}}=\oplus_{n\rr=1}^{N\rr}\mathbf{R}^{\textrm{in}}_{\textrm{r},n\rr}.
\end{align}

\subsection{IBFD MU-MIMO Communications Receiver}
\label{subsec:comm_rx}
Within the observation window, $J$ DL UEs and the BS receive both IBFD communications signals and radar probing signals. For $l$-th symbol period and $k$-th frame, denote the CSCG noise vectors measured respectively at the BS Rx and the $\ith{j}$ DL UE as $\mathbf{z}_\textrm{B}\bracket{k,l}\sim\mathcal{CN}\paren{0,\sigma^2_{\textrm{B}}\mathbf{I}_{M\cc}}$ and $\mathbf{z}_{\textrm{d},j}\bracket{k,l}\sim\mathcal{CN}\paren{0,\sigma^2_{\mathrm{d},j}\mathbf{I}_{N^\mathrm{d}_{j}}}$, i.i.d in $k$ and $l$. Then,  the signal received at the BS Rx to decode $\mathbf{s}_{\textrm{u},i}\bracket{k,l}$ and the composite signal received by the $j$-th DL UE are, respectively, 
\begin{flalign}
\mathbf{y}_{\textrm{u},i}\bracket{k,l}&=\mathbf{y}_{i,\B}\bracket{k,l}+\mathbf{y}_{\textrm{um},i}\bracket{k,l}+\mathbf{y}_{\textrm{BB}}\bracket{k,l}+\mathbf{y}_{\textrm{rB}}\bracket{k,l}+\mathbf{z}_\textrm{B}\bracket{k,l},
\label{eq: UL_total_received}\\
\textrm{and }
\mathbf{y}_{\textrm{d},j}\bracket{k,l}&=\mathbf{y}_{\textrm{B},j}\bracket{k,l}+\mathbf{y}_{\mathrm{dm},j}\bracket{k,l}+\mathbf{y}_{\mathrm{u},j}\bracket{k,l}+\mathbf{y}_{\mathrm{r},j}\bracket{k,l}+\mathbf{z}_{\textrm{d},j}\bracket{k,l}, \label{eq: DL_total_received}
\end{flalign}
\normalsize
where $\mathbf{y}_{i,\B}\bracket{k,l}$ is the signal from $i$-th UL UE to the BS; $\mathbf{y}_{\textrm{um},i}\bracket{k,l}$ is the multi-user interference (MUI) from other UL UEs; $\mathbf{y}_{\textrm{BB}}\bracket{k,l}$ is the self-interference due to FD transmission; $\mathbf{y}_{\textrm{rB}}\bracket{k,l}$ ($\mathbf{y}_{\mathrm{r},j}\bracket{k,l}$) is the radar signal received at the BS ($j$-th DL UE); $\mathbf{y}_{\textrm{B},j}\bracket{k,l}$ is the signal from the BS aimed for the $j$-th DL UE; $\mathbf{y}_{\mathrm{dm},j}\bracket{k,l}$ is the MUI from other DL UEs; and $\mathbf{y}_{\mathrm{u},j}\bracket{k,l}$ is the UL interference at the $j$-th DL UE. The CMs of $\mathbf{y}_{\textrm{u},i}\bracket{k,l}$ and $\mathbf{y}_{\textrm{d},j}\bracket{k,l}$ are, respectively, 
\begin{flalign}
\mathbf{R}_{\mathrm{u},i}\bracket{k,l}&=\mathbf{R}_{i,\B}\bracket{k,l}+ \mathbf{R}^\mathrm{in}_{\mathrm{u},i}\bracket{k,l},
\end{flalign}
and
\begin{flalign}
\mathbf{R}_{\mathrm{d},j}\bracket{k,l}&=\mathbf{R}_{\mathrm{B,j}}\bracket{k,l}+\mathbf{R}^\mathrm{in}_{\mathrm{d},j}\bracket{k,l},
\end{flalign}
where 
\begin{align}
\mathbf{R}^\mathrm{in}_{\mathrm{u},i}\bracket{k,l}=\mathbf{R}_{\textrm{um}, i}\bracket{k,l}+\mathbf{R}_{\mathrm{BB}}\bracket{k,l}+\mathbf{R}_{\textrm{rB}}\bracket{k,l}+\sigma^2_{\textrm{B}}\mathbf{I}_{\mathit{M}\cc},
\end{align}
and
\begin{align}
\mathbf{R}^{\textrm{in}}_{\textrm{d},j}\bracket{k,l}=\mathbf{R}_{\mathrm{dm},j}\bracket{k,l}+\mathbf{R}_{\mathrm{u,}j}\bracket{k,l}+\mathbf{R}_{\mathrm{r},j}\bracket{k,l}+\sigma^2_j\mathbf{I}_{N^{\textrm{d}}_{j}},
\end{align}
denote the interference-plus-noise CMs associated with \eqref{eq: DL_total_received} and \eqref{eq: UL_total_received}, respectively. The precoders of IBFD communications are based on the $\ith{n_{\mathrm{rB}}}$ symbol period of $K$ UL frames and the $\ith{n_{\textrm{rd}}}$ symbol period of $K$ DL frames, where $n_{\mathrm{rB}}$ and $n_{\textrm{rd}}$ are the symbol indices of UL and DL, respectively. 
\figurename{~\ref{fig:ReceiveSignalModel}} illustrates the composite receive signals of BS Rx, $j$-th DL UE, and $n\rr$-th radar Rx.

\begin{figure}[!t]
	\centering
	\includegraphics[width=0.80\columnwidth]{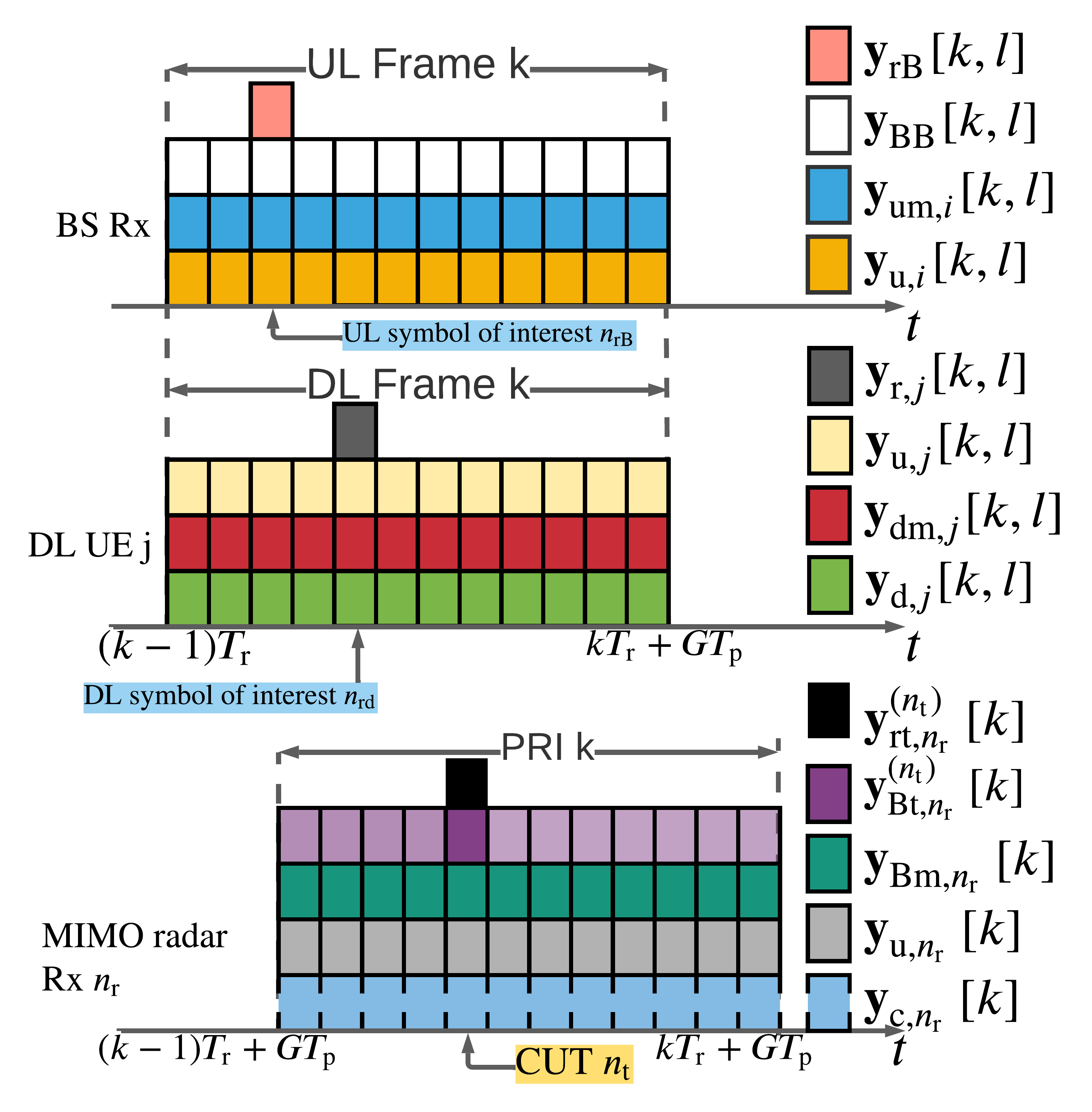}
	\caption{\textcolor{black}{The overlaid receive signal timing diagram during $\ith{k}$ radar PRI and $\ith{k}$ communications frame in the observation window; noise trails have been excluded. The purple bin with more opacity indicates the DL signal reflected from the target and observed in the radar CUT, i.e., $\mathbf{y}^{\paren{n_\target}}_{\textrm{Bt},n\rr}\bracket{k}$. } 
	} 
	\label{fig:ReceiveSignalModel}
\end{figure}

\section{CWSM Maximization}
\label{sec: formulation}
We now define the LRFs for the MIMO radar and the IBFD MU-MIMO communications system Rxs before introducing the MI-based co-design metric CWSM. Denote the LRF at the $\ith{n\rr}$ radar as $\mathbf{U}_{\mathrm{r},n\rr}=\bracket{\mathbf{u}_{\mathrm{r},n\rr}\bracket{0},\cdots,\mathbf{u}_{\mathrm{r},n\rr}\bracket{K-1}}\in\mathbb{C}^{KM\times\mathrm{\mathit{K}}}$. This LRF's output is 
\begin{align}
\widetilde{\mathbf{y}}_{\textrm{r},n\rr}&=\widetilde{\mathbf{y}}_{\mathrm{t},n\rr}+\widetilde{\mathbf{y}}^{\mathrm{in}}_{\mathrm{r},n\rr}=\mathbf{U}\rnr\mathbf{S}_{\mathrm{t},n\rr}\mathbf{h}_{\mathrm{t},n\rr}+\mathbf{U}\rnr\mathbf{y}^{\mathrm{in}}_{\mathrm{r},n\rr}
\end{align}
where $\mathbf{h}_{\mathrm{t},n\rr}\sim\mathcal{CN}\paren{\mathbf{0},\boldsymbol{\Sigma}_{\target,n\rr}}$ contains the target information and $\widetilde{\mathbf{y}}_{\mathrm{r},n\rr}\sim\mathcal{CN}\paren{\mathbf{0},\mathbf{U}\rnr\paren{\mathbf{R}_{\mathrm{t},n\rr}+\mathbf{R}^{\mathrm{in}}_{\mathrm{r},n\rr}}\mathbf{U}^\dagger\rnr}$.

Using the chain rule, the MI between $\widetilde{\mathbf{y}}_{\mathrm{r},n\rr}$ and 
$\mathbf{h}_{\mathrm{t},n\rr}$ is
\cite{Colornoise_waveform,Jammer_game}
\begin{align}
\mathit{I}\rnr&\triangleq\mathit{I}\paren{\widetilde{\mathbf{y}}\rnr;\mathbf{h}_{\mathrm{t},n\rr}}=\mathit{H}\paren{\widetilde{\mathbf{y}}\rnr}-\mathit{H}\paren{\widetilde{\mathbf{y}}\rnr|\mathbf{h}_{\mathrm{t},n\rr}}\nonumber\\
&=\mathit{H}\paren{\widetilde{\mathbf{y}}\rnr}-\mathit{H}\paren{\widetilde{\mathbf{y}}_{\textrm{t},n\rr}|\mathbf{h}_{\mathrm{t},n\rr}}-\mathit{H}\paren{\widetilde{\mathbf{y}}^{\mathrm{in}}_{\mathrm{r},n\rr}|\mathbf{h}_{\mathrm{t},n\rr}}\nonumber\\
&=\mathit{H}\paren{\widetilde{\mathbf{y}}\rnr}-\mathit{H}\paren{\widetilde{\mathbf{y}}^{\mathrm{in}}_{\mathrm{r},n\rr}},
\end{align}
where $\mathit{H}\paren{\widetilde{\mathbf{y}}_{\textrm{t},n\rr}|\mathbf{h}_{\mathrm{t},n\rr}}$ vanishes because $\widetilde{\mathbf{y}}_{\textrm{t},n\rr}$ depends on $\mathbf{h}_{\mathrm{t},n\rr}$; and $\mathit{H}\paren{\widetilde{\mathbf{y}}^{\mathrm{in}}_{\mathrm{r},n\rr}|\mathbf{h}_{\mathrm{t},n\rr}}$ reduces to $\mathit{H}\paren{\widetilde{\mathbf{y}}^{\mathrm{in}}_{\mathrm{r},n\rr}}$ because $\widetilde{\mathbf{y}}^{\textrm{in}}_{\textrm{r},n\rr}$ and $\mathbf{h}_{\textrm{t},n\rr}$ are mutually independent. 
The conditional differential entropy with the Gaussian noise \cite{Colornoise_waveform} leads to  
\begin{align}
\mathit{H}\paren{\widetilde{\mathbf{y}}\rnr|\mathbf{A}}=\varrho+\log\left|\mathbf{U}\rnr\mathbf{R}_{\mathrm{r},n\rr}\mathbf{U}^\dagger\rnr \right|
\end{align}
and
\begin{align}
\mathit{H}\paren{\widetilde{\mathbf{y}}^{\mathrm{in}}_{\mathrm{r},n\rr}|\mathbf{A}}=\varrho+\log\left|\mathbf{U}\rnr\mathbf{R}^{\mathrm{in}}_{\mathrm{r},n\rr}\mathbf{U}^\dagger\rnr \right|,
\end{align}
where the constant $\varrho=K\log\paren{\pi}+K$. This gives
\begin{align}\label{radarmi}
	\mathit{I}\rnr&=\log\frac{\left|\mathbf{U}\rnr\paren{\mathbf{R}_{\mathrm{t},n\rr}+\mathbf{R}^{\mathrm{in}}_{\mathrm{r},n\rr}}\mathbf{U}^\dagger\rnr \right|}{\left|\mathbf{U}\rnr\mathbf{R}^{\mathrm{in}}_{\mathrm{r},n\rr}\mathbf{U}^\dagger\rnr \right|}\nonumber\\
	&=\log\left\lvert\mathbf{I}+\mathbf{U}\rnr\mathbf{R}_{\textrm{t},n\rr}\mathbf{U}^\dagger\rnr\paren{\mathbf{U}\rnr\mathbf{R}^{\mathrm{in}}_{\mathrm{r},n\rr}\mathbf{U}^\dagger\rnr}^{-1}\right\rvert.
\end{align}

The LRFs deployed at the BS to decode the $\ith{i}$ UL UE and $\ith{j}$ DL UE during $\ith{k}$ frame of the observation window are  $\UiB\in\mathbb{C}^{\mathit{D}^{\textrm{u}}_{i}\times N\cc}$ and $\UBj\in\mathbb{C}^{\mathit{D}^{\textrm{d}}_{j}\times N^{\textrm{d}}_{j}}$, respectively. The outputs of $\UiB$ and $\UBj$ are $\widetilde{\mathbf{y}}_{\mathrm{u},i}\bracket{k,l}=\UiB\mathbf{y}_{\mathrm{u},i}\bracket{k,l}$ and $\widetilde{\mathbf{y}}_{\mathrm{d},j}\bracket{k,l}=\UBj\mathbf{y}_{\mathrm{d},j}\bracket{k,l}$; these signals follow the distributions $\mathcal{CN}\paren{\mathbf{0},\UiB\mathbf{R}_{\mathrm{u},i}\bracket{k,l}\UiBH}$ and $\mathcal{CN}\paren{\mathbf{0},\UBj\mathbf{R}_{\mathrm{d},j}\bracket{k,l}\UBjH}$, respectively. 

Using the common assumption of Gaussianity on symbol vectors, i.e. $\mathbf{s}_{\mathrm{u},i}\bracket{k,l}\sim\mathcal{CN}\paren{\mathbf{0},\PiB\PiBH}$ and $\mathbf{s}_{\mathrm{d},j}\bracket{k,l}\sim\mathcal{CN}\paren{\mathbf{0},\PBj\PBjH}$, the MIs between $\widetilde{\mathbf{y}}_{\mathrm{u},i}\bracket{k,l}$ and $\mathbf{s}_{\mathrm{u},i}\bracket{k,l}$ as well as $\widetilde{\mathbf{y}}_{\mathrm{d},j}\bracket{k,l}$ and $\mathbf{s}_{\mathrm{d},j}\bracket{k,l}$ are
\begin{flalign}
\mathit{I}^{\textrm{u}}_{i}\bracket{k,l}&\triangleq \mathit{I}\paren{\mathbf{s}_{\mathrm{u},i}\bracket{k,l};\widetilde{\mathbf{y}}_{\mathrm{u},i}\bracket{k,l}}
=\log\left|\mathbf{I}+\UiB\mathbf{R}_{i,\B}\bracket{k,l}\UiBH\paren{\UiB\mathbf{R}^{\mathrm{in}}_{\mathrm{u},i}\bracket{k,l}\UiBH}^{-1}\right|,
\end{flalign}
and
\begin{flalign}
\mathit{I}^{\textrm{d}}_{j}\bracket{k,l}&\triangleq I\paren{\mathbf{s}_{\mathrm{B},j}\bracket{k,l};\widetilde{\mathbf{y}}_{\mathrm{d},j}\bracket{k,l}}=\label{DLmutual} =\log\left|\mathbf{I}+\UBj\mathbf{R}_{\textrm{B},j}\bracket{k,l}\UBjH\paren{\UBj\mathbf{R}^{\mathrm{in}}_{\mathrm{d},j}\bracket{k,l}\UBjH}^{-1}\right|,
\end{flalign}
respectively. We evaluate FD communications based on the symbols-of-interest, i.e., $l=n_{\mathrm{rB}}$ for each UL frame and $l=n_{\textrm{rd}}$ for each DL frame in the observation window. The metric CWSM is a weighted sum of communications' MIs related to the symbol periods of interest\footnote{Hereafter, for simplicity, we drop symbol index $l=n_{\textrm{rB}}$ ($l=n_{\textrm{rd}}$) for UL (DL) related terms.} and $I_{\mathrm{r},n\rr}$, i.e., 
\begin{flalign}
\label{objectfunction1}
\mathit{I}_{\textrm{CWSM}}&=\sum_{n\rr=1}^{N\rr}\alpha^\textrm{r}_{n\rr} \mathit{I}^{\textrm{r}}_{n\rr}+\sum_{k=0}^{K-1}\bracket{\sum_{i=1}^\mathit{I}\alpha^\textrm{u}_i\mathit{I}^{\textrm{u}}_{i}\bracket{k}+\sum_{j=1}^\mathit{J}\alpha^\textrm{d}_j\mathit{I}^{\textrm{d}}_{j}\bracket{k}},
\end{flalign}
where $\alpha^\textrm{r}_{n\rr}$, $\alpha^\textrm{u}_i$, and $\alpha^\textrm{d}_j$ are pre-defined weights assigned to the MIMO $n\rr$-th radar Rx , $\ith{i}$ UL UE and $\ith{j}$ DL UE, respectively, for all $n\rr$, $i$, and $j$; the weights are determined by the system priority and specific applications. For example, for FD communications, the weights are based on available buffer capacities of BS and UEs\cite{MSE_FD,FD_WMMSE}. For the joint radar-communications, one can assign larger (smaller) weights to $\alpha^{\mathrm{r}}_{n_\mathrm{r}}$ with the presence (absence) of targets 
\cite{Liu2018Gloabalsip}.

Denote the sets of the precoders and the LRFs as $\braces{\mathbf{P}}\triangleq\braces{\PiB,\PBj | i\in\mathbb{Z}_+\braces{I}, j\in\mathbb{Z}_+\braces{J}, k\in\mathbb{Z}_+\braces{K}}$
and
$\braces{\mathbf{U}}\triangleq\left\lbrace\UiB, \UBj, \mathbf{U}\rnr | i\in\mathbb{Z}_+\braces{I}, j\in\mathbb{Z}_+\braces{J},\right.$ $\left.k\in\mathbb{Z}_+\braces{K}, n\rr\in \mathbb{Z}_+\braces{N\rr}\right\rbrace$. 
The transmission powers that occurred to the BS and the $i$-th UL UE at the $\ith{k}$ frame are
\begin{align}
\label{eq: DL_power}
P_{\mathrm{d}}\bracket{k}=\sum_{j=1}^{J}P_{\mathrm{d},j}\bracket{k}=\sum_{j=1}^{J}\trace\braces{\PBj\PBjH},
\end{align}
and 
\begin{align}
P_{\mathrm{u},i}\bracket{k}=\trace\braces{\PiB\PiBH}, \label{eq: UL_power}
\end{align}
which are upper bounded by the maximum DL and UL powers $P_\B$ and $P_{\mathrm{U}}$, respectively. The achievable rates for the $i$-th UL UE and the $j$-th DL UE in the $\ith{k}$ frame, $R_{\mathrm{u},i}\bracket{k}$ and $R_{\mathrm{d},j}\bracket{k}$ are lower bounded by the least acceptable achievable rates to quantify the QoS of the UL and DL, $\mathit{R}_{\textrm{UL}}$ and $\mathit{R}_{\textrm{DL}}$, respectively. The CWSM optimization to jointly design precoders $\braces{\mathbf{P}}$, radar code $\mathbf{A}$, and LRFs $\braces{\mathbf{U}}$ is 
\begin{subequations}\label{jointop}\begin{align}
	\underset{{\braces{\mathbf{P}},\braces{\mathbf{U}},
			\mathbf{A}}}{\text{maximize}}\;& \mathit{I}_{\textrm{CWSM}}\paren{\braces{\mathbf{U}},\braces{\mathbf{P}},\mathbf{A}}  \\
	\text{subject to}\; & P_{\mathrm{d}}\bracket{k}\leq             
                            P_\textrm{B},\label{DL_power}\\
                            &P_{\mathrm{u},i}\bracket{k}\leq P_\textrm{U}, \\*
	                       &\mathit{R}_{\textrm{u},i}\bracket{k}\geq\mathit{R}_{\textrm{UL}}, \\
                           &\mathit{R}_{\textrm{d},j}\bracket{k}\geq \mathit{R}_{\textrm{DL}},\label{DLrate}\\
	                   &\lVert\mathbf{a}_{m\rr}\rVert^2 =P_{\textrm{r},m\rr},\; \label{constraint:radarpower}\\*
	                    &\frac{\mathrm{\mathit{K}}\max_{k=1,\cdots,  K}\lvert\mathbf{a}_{m\rr}\bracket{k}\rvert^2}{P_{\mathrm{r},m\rr}}\leq\mathrm{\gamma}_{m\rr},\; \forall\; i,j,k,m\rr,\label{constraint:radarpar}
\end{align}
\end{subequations}
where constraints \eqref{constraint:radarpower} and \eqref{constraint:radarpar}  are determined by the transmit power and PAR of the $\ith{m\rr}$ MIMO radar Tx, respectively. Note that PAR constraint is applied column-wise to the code matrix $\mathbf{A}$ because the Txs of a statistical MIMO radar are widely distributed. When $\gamma_{m\rr}$ = 1, PAR constraint is reduced to constant modulus constraint. 

\section{Joint Code-Precoder-Filter Design}
\label{sec:solution}
Even without the non-convex constraints $\eqref{DLrate}$-$\paren{\mathrm{\ref{constraint:radarpar}}}$, $\paren{\mathrm{\ref{jointop}}}$ is non-convex because the objective function $\mathit{I}_{\textrm{CWSM}}$ is not jointly concave over ${\braces{\mathbf{P}}}$, $\braces{\mathbf{U}}$, and $\mathbf{A}$, and therefore its global optima are generally intractable \cite{Lui2006subg}. In general, such a problem is solved by alternately optimizing over one unknown variable at a time. When the number of variables is large, methods such as the BCD partition all optimization variables into, say, $V$ small groups or blocks and optimize over each block, one at a time, while keeping other block variables fixed \cite{BCDconvergence}. The net effect is that the problem is equivalently solved by iteratively solving less complex $V$ subproblems. If there are only two blocks of variables, the BCD reduces to the classical alternating minimization method\cite{BCDconvergence,Liu2017asilomar}.  

It has been shown \cite{ADMMBCD,zhang2017convergent} that the BCD converges globally to a stationary point for both convex and non-convex problems while methods such as Alternating Direction Method of Multipliers (ADMM) and Douglas-Rachford Splitting (DRS) achieve only linear convergence for strictly convex and some non-convex (e.g. multi-convex) problems. The stochastic gradient descent used to address saddle point problems has a slower convergence rate than BCD and offers only weak convergence for non-convex problems \cite{zhang2017convergent}. 

One can partition the block coordinate variables from \eqref{jointop} into three groups, i.e., $\braces{\mathbf{P}}$, $\mathbf{A}$, and $\braces{\mathbf{U}}$. At each iteration, we apply a \textit{direct update} \cite{BCDconvergence}, i.e., maximize $I_{\textrm{CWSM}}$ for all the block variables. Further, we update the block variables in a cyclic sequence because its global and local convergence has been well-established \cite{BCDconvergence,Lops2019serveillance} 
compared to other sequential update rules\footnote{A possible alternative is the randomized BCD, where the series of iterates generated by BCD are divergent. However, cyclic BCD may still outperform the randomized BCD \cite{ADMMBCD}.}. In particular, we adopt the Gauss-Seidel BCD \cite{BCDconvergence}, which minimizes the objective function cyclically over each block while keeping the other blocks fixed. 

A summary of our strategy is as follows. Note that ${\braces{\mathbf{P}}}$ is subject to only communications-centric constraints $\paren{\mathrm{\ref{DL_power}}}$-$\paren{\mathrm{\ref{DLrate}}}$ and $\mathbf{A}$ to both radar-centric PAR constraints $\paren{\mathrm{\ref{constraint:radarpower}}}$-$\paren{\mathrm{\ref{constraint:radarpar}}}$ and communications-centric constraints $\paren{\mathrm{\ref{DL_power}}}$-$\paren{\mathrm{\ref{DLrate}}}$. 
We denote the sets of feasible $\mathbf{A}$ determined by $\paren{\mathrm{\ref{DL_power}}}$-$\paren{\mathrm{\ref{DLrate}}}$ and $\paren{\mathrm{\ref{constraint:radarpower}}}$-$\paren{\mathrm{\ref{constraint:radarpar}}}$ as $\mathbb{A}_{\textrm{c}}$ and $\mathbb{A}_{\textrm{r}}$, respectively and the optimal solution for $\mathbf{A}$, i.e., $\mathbf{A}^\star$ is thus in the intersection of $\mathbb{A}_{\textrm{c}}$ and $\mathbb{A}_{\textrm{r}}$, i.e., $\mathbf{A}^\star\in\mathbb{A}_{\textrm{c}}\cap\mathbb{A}_{\textrm{r}}$. The AP method \cite{arXiv180203889Z,nearestvector} is appropriate to perform the search for $\mathbf{A}^\star$. In the sequel of this section, we first transform the $I_{\textrm{CWSM}}$ maximization problem, which is a weighted sum rate (WSR) problem,  in \eqref{jointop} to a weighted minimum mean-squared-error (WMMSE) minimization problem without the PAR constraints in Section \ref{subsec: MMSE section}. It has been shown in \cite{Luo2011IterativeWMMSE,FD_WMMSE} that maximizing information-theoretic quantities via WMMSE for MIMO precoder design yields better results than geometric programming. Mapping a WSR problem to its corresponding WMMSE problem is also more computationally efficient than gradient-based (GB) approaches \cite{FD_WMMSE} because of its low per-iteration complexity, which is guaranteed to converge to at least a local optimum \cite{Luo2011IterativeWMMSE}. In Section~\ref{subsec:seq}, we then employ the BCD in our proposed WMMSE-MRMC algorithm to solve for the optimal $\mathbf{A}$ in $\mathbb{A}_{\textrm{c}}$, which we denote as $\mathbf{A}^\prime$, and the optimal $\braces{\mathbf{P}^\star}$. Then, we project each column of $\mathbf{A}^\prime$ onto $\mathbb{A}_{\textrm{r}}$ in Section~\ref{subsec: PAR} using AP. The WMMSE-MRMC and AP procedures comprise the overall BCD-AP MRMC algorithm (Algorithm~\ref{Alternating_sum}) and are repeated until convergence.

\subsection{Relationship between WMMSE and MI}
\label{subsec: MMSE section}
Consider $\paren{\mathrm{\ref{jointop}}}$ without PAR constraints,
\begin{equation}\label{jointop_first}
\underset{{\braces{\mathbf{P}},\braces{\mathbf{U}},
\mathbf{A}}}{\text{maximize}}\; \mathit{I}_{\textrm{CWSM}}\paren{\braces{\mathbf{U}},\braces{\mathbf{P}},\mathbf{A}}\;\;\text{subject to}\; 
\eqref{DL_power}-\eqref{DLrate}.
\end{equation}
To derive the WMMSE expressions regarding $\eqref{jointop_first}$, we first define the mean squared error for the $\ith{n\rr}$ radar Rx, $\ith{i}$ UL UE, and $\ith{j}$ DL UE as 
\begin{flalign}
\label{radarMSE}
\mathbf{E}\rnr&=\mathbb{E}\bracket{\paren{\mathbf{h}_{\mathrm{t},n\rr}-\mathbf{U}\rnr\mathbf{y}_{\mathrm{r},n\rr}}\paren{\mathbf{h}_{\mathrm{t},n\rr}-\mathbf{U}\rnr\mathbf{y}_{\mathrm{r},n\rr}}^\dagger}\nonumber\\
&=\boldsymbol{\Sigma}_{\target,n\rr}-\mathbf{U}\rnr\mathbf{S}_{\target,n\rr}\boldsymbol{\Sigma}_{\target,n\rr}-\boldsymbol{\Sigma}^\dagger_{\target,n\rr}\mathbf{S}^\dagger_{\target,n\rr}\mathbf{U}^\dagger\rnr\nonumber\\
&+\mathbf{U}\rnr\mathbf{R}_{\textrm{r},n\rr}\mathbf{U}^\dagger\rnr,
\end{flalign}
\begin{flalign}
\label{ULMSE}
\EiB &=\mathbb{E}\bracket{\paren{\duis-\UiB\yui}\paren{\duis-\UiB\yui}^\dagger}\nonumber\\
&=\mathbf{I}-\UiB\HiB\PiB-\PiBH\HiBH\UiBH\nonumber\\
&+\UiB\Ris\UiBH,
\end{flalign}
\begin{flalign}
\label{DLMSE}
\EBj&=\mathbb{E}\bracket{\paren{\ddjs-\UBj\ydj}\paren{\ddjs-\UBj\ydj}^\dagger}\nonumber\\
&=\mathbf{I}-\UBj\HBj\PBj-\PBjH\HBjH\UBjH\nonumber\\
&+\UBj\Rjs\UBjH,
\end{flalign}\normalsize
where the expectations are taken w.r.t.  $\mathbf{h}_{\mathrm{t},n\rr}$, $\duis$, and $\ddjs$, respectively. Denote symmetric weight matrices associated with $\Ernr$, $\EiB$, and $\EBj$ as $\mathbf{W}\rnr\in\mathbb{C}^{KM\times KM}\succeq\mathbf{0}$, $\WiB\in\mathbb{C}^{D^{\mathrm{u}}_{i}\times D^{\mathrm{u}}_{i}}\succeq\mathbf{0}$, and $\WBj\in\mathbb{C}^{D^{\mathrm{d}}_{j}\times D^{\mathrm{d}}_{j}}\succeq\mathbf{0}$, respectively. The weighted-sum MSE is 
\begin{flalign}\label{Xi_Mses}
&\Xi_{\textrm{wmse}}\triangleq\underbrace{\sum_{k=1}^{\mathrm{\mathit{K}}}\sum_{i=1}^\mathit{I}\alpha^\textrm{u}_i\trace\braces{\WiBn\EiBn}}_{=\Xi_{\textrm{UL}}}+\underbrace{\sum_{n\rr=1}^{N\rr}\alpha^\textrm{r}_{n\rr}\trace\braces{\mathbf{W}\rnr\mathbf{E}\rnr}}_{=\Xi_{\textrm{r}}} \nonumber\\
& +\underbrace{\sum_{k=1}^{\mathrm{\mathit{K}}}\sum_{j=1}^\mathit{J}\alpha^\textrm{d}_j\trace\braces{\WBjone\EBjone} }_{=\Xi_{\textrm{DL}}},
\end{flalign}
where $\WiB\EiB$, $\WBj\EBj$ and $\Wrnr\Ernr$ are: 
\begin{align}
	&\WiB\EiB\nonumber\\
 &=\WiB\mathbf{I}-\WiB\UiB\HiB\PiB-\WiB\PiBH\HiBH\UiBH+\WiB\UiB\HiB\PiB\PiBH\HiBH\UiBH\nonumber\\
	&+\WiB\UiB\paren{\sum_{q\neq i}\HqB\PqB\PqBH\HqBH+\sum_{j=1}^{\mathit{J}}\HBB\PBj\PBjH\HBBH+\HrB\mathbf{a}\bracket{k}\mathbf{a}^\dagger\bracket{k}\HrBH}\UiBH,	\label{WEi}\\
	&\WBj\EBj\nonumber\\
 &=\WBj\mathbf{I}-\WBj\UBj\HBj\PBj-\WBj\PBjH\HBjH\UBjH\nonumber\\
 &+\WBj\UBj\HBj\PBj\PBjH\HBjH\UBjH\nonumber\\
	&+\WBj\UBj\paren{\sum_{g\neq j}\HBj\PqB\PqBH\HBjH+\sum_{i=1}^{\mathit{I}}\Hij\PiB\PiBH\HijH+\Hrj\mathbf{a}\bracket{k}\mathbf{a}^\dagger\bracket{k}\HrjH}\UBjH,		\label{WEj}
 \end{align}
and
 \begin{align}
&\mathbf{W}\rnr\mathbf{E}\rnr\nonumber\\
&=\mathbf{W}\rnr\mathbb{E}\bracket{\paren{\mathbf{h}_{\mathrm{t},n\rr}-\sum_{k=1}^{\mathrm{\mathit{K}}}\urk\braces{\mathbf{y}				_{\textrm{t},n\rr}\bracket{k}+\mathbf{y}^{\mathrm{in}}_{\mathrm{r},n\rr}\bracket{k}}}\paren{\mathbf{h}_{\mathrm{t},n\rr}-\sum_{k=1}^{\mathrm{\mathit{K}}}\urk\braces{\mathbf{y}					_{\textrm{t},n\rr}\bracket{k}+\mathbf{y}^{\mathrm{in}}_{\mathrm{r},n\rr}\bracket{k}}}^\dagger}\nonumber\\
	&=\Wrnr\boldsymbol{\Sigma}_{\target,n\rr}-\Wrnr\mathbb{E}\bracket{\sum_{m=1}^{K}\mathbf{J}_{\textrm{h}}\bracket{m}\paren{\mathbf{J}_{\textrm{r}}\mathbf{h}_{\textrm{rt},n\rr}\bracket{m}+\mathbf{J}_{\textrm{B}}\mathbf{h}_{\textrm{Bt},n\rr}\bracket{m}}\sum_{\ell=1}^K\paren{\mathbf{a}^\dagger\bracket{\ell}\mathbf{h}^\ast_{\mathrm{rt},n\rr}\bracket{\ell}+\mathbf{s}^\dagger_{\mathrm{Bt},n\rr}\bracket{\ell}\mathbf{h}^\ast_{\mathrm{Bt},n\rr}\bracket{\ell}}\mathbf{u}^\dagger\rnr\bracket{\ell}}\nonumber\\
	&-\Wrnr\mathbb{E}\bracket{\sum_{m=1}^{K}\mathbf{u}\rnr\bracket{m}\paren{\mathbf{h}_{\mathrm{rt},n\rr}\bracket{m}\mathbf{a}\bracket{m}+\mathbf{h}_{\mathrm{Bt},n\rr}\bracket{m}\mathbf{s}_{\mathrm{Bt},n\rr}\bracket{m}}\sum_{\ell=1}^K\paren{\mathbf{h}^\dagger_{\textrm{rt},n\rr}\bracket{\ell}\mathbf{J}^\dagger_{\textrm{r}}+\mathbf{h}^\dagger_{\textrm{Bt},n\rr}\bracket{\ell}\mathbf{J}^\dagger_{\textrm{B}}}\mathbf{J}^\dagger_{\textrm{h}}\bracket{\ell}}\nonumber\\
	&+\Wrnr\sum_{m=1}^{\mathrm{\mathit{K}}}\sum_{\ell=1}^{\mathrm{\mathit{K}}}\mathbf{u}\rnr\bracket{m}\paren{\mathbf{R}_{\mathrm{rt},n\rr}\paren{m,\ell}+\mathbf{R}_{\mathrm{Bt},n\rr}\paren{m,\ell}+\mathbf{R}_{\textrm{Bm},n\rr}\paren{m,\ell}+\mathbf{R}_{\textrm{U},n\rr}\paren{m,\ell}+\mathbf{R}_{\textrm{c},n\rr}\paren{m,\ell}+\sigma_{n\rr}}\mathbf{u}^\dagger\rnr\bracket{\ell}.
	\label{WEr}
	\end{align}\normalsize
\color{black}
Minimizing \eqref{Xi_Mses} is the key to solving the problematic non-convex problem in \eqref{jointop_first}, as stated in the following theorem.

\begin{theorem}
\label{the: one}
Solving the problem 
\begin{align}
	\label{WMMSE1}
	\underset{{\braces{\mathbf{P}},\braces{\mathbf{U}},
			\mathbf{A},\braces{\mathbf{W}}}}{\textrm{minimize}}&\quad\Xi_{\textrm{wmse}}\paren{\braces{\mathbf{U}}, \braces{\mathbf{P}},\mathbf{A},\braces{\mathbf{W}}},\\
	\textrm{subject to}&\quad \eqref{DL_power}-\eqref{DLrate}\nonumber 
\end{align} \normalsize
yields the exact solution of the problem \eqref{jointop_first}.
\end{theorem}
\begin{proof}
The optimization of $\paren{\ref{WMMSE1}}$ w.r.t. $\braces{\mathbf{U}}$ yields \cite{FD_WMMSE}
\begin{flalign}
\label{radarWMMSE_Rx}
\mathbf{U}^\star\rnr&=\arg\min_{\mathbf{U}\rnr,\forall n\rr}\trace\braces{\mathbf{W}\rnr\mathbf{E}\rnr}\nonumber\\
&\;\;\;\;\;\;\;\;\;\;=\boldsymbol{\Sigma}_{\target,n\rr}\mathbf{S}^\dagger_{\target,n\rr}\paren{\mathbf{S}_{\mathrm{t},n\rr}\boldsymbol{\Sigma}_{\target,n\rr}\mathbf{S}^\dagger_{\target,n\rr}+\mathbf{R}^{\mathrm{in}}_{\mathrm{r},n\rr}}^{-1},\\
\mathbf{U}^\star_{\textrm{u},i}\bracket{k}&=\arg\min_{\UiB,\forall i,k,l}\trace\braces{\WiB\EiB}\nonumber\\
&\;\;\;\;\;\;\;\;\;\;\;=\PiBH\HiBH\paren{\mathbf{R}_{\textrm{u},i}\bracket{k}}^{-1},\label{UL_WMMSE_Rx}\\
\text{and }\mathbf{U}^\star_{\textrm{d},j}\bracket{k}&=\arg\min_{\UBj,\forall j,k,l}\trace\braces{\WBj\mathbf{E}_{\textrm{d},j}\bracket{k}}\nonumber\\
&\;\;\;\;\;\;\;\;\;\;\;=\PBjH\mathbf{H}^\dagger_{\textrm{B},j}\paren{\Rjs}^{-1}.\label{DL_WMMSE_Rx}
\end{flalign}\normalsize
Substituting $\mathbf{U}^\star\rnr$, $\mathbf{U}^\star_{\textrm{u},i}\bracket{k}$, and $\mathbf{U}^\star_{\textrm{d},j}\bracket{k}$ into MIs in $\paren{\ref{radarmi}}$-$\paren{\ref{DLmutual}}$ and MSEs in $\paren{\ref{radarMSE}}$-$\paren{\ref{DLMSE}}$ yields, respectively, the achievable rates of $\ith{n\rr}$ radar Rx, $\ith{i}$ UL UE, and $\ith{j}$ DL UE as
\begin{align}
\mathit{R}_{\textrm{r},n\rr}&=	\log\left|\mathbf{I}+\mathbf{S}_{\mathrm{t},n\rr}\boldsymbol{\Sigma}_{\target,n\rr}\mathbf{S}^\dagger_{\target,n\rr}\mathbf{R}^{-1}_{\mathrm{in,n\rr}}\right|,\\
\mathit{R}_{\textrm{u},i}\bracket{k}&=\log\left|\mathbf{I}+\mathbf{R}_{i,\B}\bracket{k}\Riniin\right|, \label{UL_rate}
\end{align}
and
\begin{align}
\mathit{R}_{\textrm{d},j}\bracket{k}&=\log\left|\mathbf{I}+\mathbf{R}_{\textrm{B},j}\bracket{k}\Rinjins\right| \label{DL_rate}.
\end{align}
The corresponding minimum mean-squared-error (MMSE) of $\ith{n\rr}$ radar Rx, $\ith{i}$ UL UE, and $\ith{j}$ DL UE are, respectively,
\begin{align}
\mathbf{E}^{\star}\rnr&=\boldsymbol{\Sigma}_{\target,n\rr}\bracket{\mathbf{I}-\mathbf{S}^\dagger_{\target,n\rr}\paren{\mathbf{R}_{\mathrm{r},n\rr}}^{-1}\mathbf{S}_{\target,n\rr}\boldsymbol{\Sigma}_{\target,n\rr}},\label{radarMMSE} \\
\mathbf{E}^{\star}_{\mathrm{u},i}\bracket{k}&=\mathbf{I}-\PiBH\mathbf{H}^\dagger_{i,\textrm{B}}\paren{\Ri}^{-1}\mathbf{H}_{i,\textrm{B}}\PiB,\label{ULMMSE}
\end{align}
and
\begin{align}
\mathbf{E}^{\star}_{\mathrm{d},j}\bracket{k}&=\mathbf{I}-\PBjH\mathbf{H}^\dagger_{\textrm{B},j}\paren{\Rj}^{-1}\mathbf{H}_{\textrm{B},j}\PBj.\label{DLMMSE}
\end{align}
The \textit{data processing inequality} \cite[p.34]{cover2006elements} implies that $\mathit{R}_{\mathrm{r}, n\rr}$, $\mathit{R}_{\textrm{u},i}\bracket{k,l}$, and $\mathit{R}_{\textrm{d},j}\bracket{k,l}$ are the upper bounds of $\mathit{I}_{\mathrm{r},n\rr}$,  $\mathit{I}^{\textrm{u}}_{i}\bracket{k,l}$, and $\mathit{I}^{\textrm{d}}_{j}\bracket{k,l}$, for all $n\rr$, $i$, and $j$. It follows that
\begin{align}
\braces{\mathbf{U}^\star}\triangleq\braces{\mathbf{U}^\star\rnr, \mathbf{U}^\star_{\textrm{u},i}\bracket{k},\mathbf{U}^\star_{\textrm{d},j}\bracket{k},\forall \braces{n\rr,i,j}},
\end{align}
is also the optimal solution of \eqref{jointop_first} and, in turn, the original problem \eqref{jointop}, whose additional constraints do not affect the solution for $\braces{\mathbf{U}}$.  Using Woodbury matrix identity \cite{MSE_FD}, the achievable rates are  
\begin{flalign}
&\mathit{R}_{\mathrm{r},n\rr}=\log\left|\boldsymbol{\Sigma}_{\target,n\rr}\paren{\mathbf{E}^{\star}\rnr}^{-1}\right|, \;\mathit{R}_{\textrm{u},i}\bracket{k}=\log\left|\paren{\mathbf{E}^{\star}_{\mathrm{u},i}\bracket{k}}^{-1} \right|,\nonumber\\
\text{and }&\mathit{R}_{\textrm{d},j}\bracket{k}=\log\left| \paren{\mathbf{E}^{\star}_{\mathrm{d},j}\bracket{k}}^{-1}\right|.
\end{flalign}\normalsize
Applying the first order optimal condition \cite{Luo2011IterativeWMMSE,MSE_FD} w.r.t. $\braces{\mathbf{W}}$ produces the optimal weight matrices $\braces{\mathbf{W}^\star}$ as
\begin{align}
\mathbf{W}^\star_{\textrm{u},i}\bracket{k}=\paren{\mathbf{E}^\star_{\textrm{u},i}\bracket{k}}^{-1},\nonumber\\ \WBjop=\paren{\mathbf{E}^\star_{\textrm{d},j}\bracket{k}}^{-1},
\end{align}
and 
\begin{align}
\mathbf{W}^\star\rnr=\paren{\mathbf{E}^\star\rnr}^{-1},
\end{align}
for all $\braces{n\rr,i,j}$. Define
\begin{flalign}
&\Xi_{\text{wmse}}^\prime=\Xi_{\text{wmse}}-\sum_{n\rr=1}^{N\rr}\alpha^\textrm{r}_{n\rr}\paren{\boldsymbol{\Sigma}_{\target,n\rr}\log\left| \mathbf{W}\rnr\right|+\mathit{KM}}-\nonumber\\
&\sum_{k=0}^{K-1}\braces{\sum_{i=1}^\mathit{I}\alpha^\textrm{u}_i\paren{\log\left| \WiB\right|+\mathit{D}^{\textrm{u}}_{i}}+ \sum_{j=1}^\mathit{J}\alpha^\textrm{d}_j\paren{\log\left| \WBj\right|+\mathit{D}^{\textrm{d}}_{j}}}.\nonumber
\end{flalign}\normalsize
Substituting $\braces{\mathbf{W}^\star}$ and $\braces{\mathbf{U}^\star}$ in $\Xi_{\text{wmse}}^\prime$ results in
\begin{align}
&\Xi^\prime_{\textrm{wmse}}=
\sum_{k=0}^{K-1}\left\lbrace
\sizecorr{\sum_{i=1}^I\alpha^\textrm{u}_i\paren{\log\left| \WiB\right|-\mathit{D}^{\textrm{u}}_{i}}}
-\sum_{j=1}^{\mathit{J}}\alpha^\textrm{d}_j\log\left|\paren{\mathbf{E}^\star_{\mathrm{d},j}\bracket{k}}^{-1}\right|\right.\nonumber\\
&-\left.\sum_{i=1}^{\mathit{I}}\alpha^\textrm{u}_i\log\left|\paren{\mathbf{E}^{\star}_{\mathrm{u},i}\bracket{k}}^{-1}\right|\right\rbrace -\alpha\rr\sum_{n\rr=1}^{N\rr}\log\left|\paren{\mathbf{E}^\star\rnr}^{-1}\right|\nonumber\\
&=\sum_{k=1}^K\braces{\sum_{j=1}^J\alpha^\textrm{d}_j\mathit{R}_{\textrm{d},j}\bracket{k}+\sum_{i=1}^I\alpha^\textrm{u}_i\mathit{R}_{\textrm{u},i}\bracket{k}}+\sum_{n\rr=1}^{N\rr}\alpha^\textrm{r}_{n\rr} \mathit{R}\rnr\nonumber\\
&=-\mathit{I}_{\textrm{CWSM}}\left(\braces{\mathbf{U}^\star},\braces{\mathbf{P}},\mathbf{A}\right),
\end{align}\normalsize
which indicates that maximizing $\mathit{I}_{\textrm{CWSM}}$ is equivalent to minimizing $\Xi^\prime_{\textrm{wmse}}$ given  $\braces{\mathbf{U}^\star}$. As minimizing $\Xi^\prime_{\textrm{wmse}}$ w.r.t. $\braces{\mathbf{P}}$ and $\mathbf{A}$ is also equivalent to minimizing $\Xi_{\textrm{wmse}}$ given $\braces{\mathbf{U}^\star}$ and $\braces{\mathbf{W}^\star}$. 
This completes the proof. 
\end{proof}	
Substituting $\braces{\mathbf{U}^\star}$ and $\braces{\mathbf{W}^\star}$ in $\paren{\ref{WMMSE1}}$ yields the WMMSE 
\begin{align}
\Xi_{\textrm{wmmse}}\paren{\braces{\mathbf{P}},\mathbf{A}}\triangleq\Xi_{\textrm{wmse}}\paren{\braces{\mathbf{U}^\star}, \braces{\mathbf{W}^\star} \braces{\mathbf{P}}}.
\end{align}
Therefore, $\mathbf{A}$ and $\braces{\mathbf{P}}$ are obtained by solving 
\begin{equation}
\label{WMMSE2}
\underset{{\braces{\mathbf{P}},\mathbf{A}}}{\textrm{minimize}}\quad\Xi_{\textrm{wmmse}}\paren{\braces{\mathbf{P}},\mathbf{A}}\;\textrm{subject to}\quad \eqref{DL_power}-\eqref{DLrate}. 
\end{equation}\normalsize

\subsection{WMMSE-MRMC} \label{subsec:seq}
In order to solve \eqref{WMMSE2}, we sequentially iterate over each element in $\braces{\mathbf{P}}$ and each row of $\mathbf{A}$, i.e., $\mathbf{a}^\top\bracket{k}$,
using the Lagrange dual method to find a closed-form solution to each variable, which constitutes our WMMSE-MRMC algorithm. 
 
\subsubsection{Lagrange dual solution}Denote the Lagrange multiplier vectors w.r.t. constraints \eqref{DL_power}-\eqref{DLrate}, respectively, as 
\begin{align}
\boldsymbol{\lambda}_{\textrm{DL}} &\triangleq\bracket{\lambda_\textrm{d}\bracket{0},\cdots,\lambda_\textrm{d}\bracket{K-1}}^\top\in\mathbb{R}^{K},\\
\boldsymbol{\lambda}_{\textrm{UL}} &\triangleq\bracket{\lambda_{\textrm{u},1}\bracket{0},\cdots,\lambda_{\textrm{u},I}\bracket{K-1}}^\top\in\mathbb{R}^{KI},\\ \boldsymbol{\mu}_{\textrm{DL}} &\triangleq\bracket{\mu_{\textrm{d},1}\bracket{0},\cdots,\mu_{\textrm{d},J}\bracket{K-1}}^\top\in\mathbb{R}^{KJ},
\end{align}
and 
\begin{align}
\boldsymbol{\mu}_{\text{UL}} &\triangleq\bracket{\mu_{\textrm{u},1}\bracket{0},\cdots,\mu_{\textrm{u},I}\bracket{K-1}}^\top\in\mathbb{R}^{KI},
\end{align}
as well as the UL power vector, the DL power vector, the UL rate vector, and the DL rate vector as, respectively,
\begin{align}
\mathbf{p}_{\textrm{UL}} &\triangleq\bracket{P_{\textrm{u},1}\bracket{0},\cdots,P_{\textrm{u,}I}\bracket{K-1}}^\top\in\mathbb{R}^{KI}, \\ 
\mathbf{p}_{\textrm{DL}} &\triangleq \bracket{P_{\mathrm{d}}\bracket{0},\cdots,P_{\mathrm{d}}\bracket{K-1}}^\top\in\mathbb{R}^{K}, \\
\mathbf{r}_{\textrm{UL}} &\triangleq\bracket{\mathit{R}_{\textrm{u},1}\bracket{0},\cdots,\mathit{R}_{\textrm{u},I}\bracket{K-1}}^\top\in\mathbb{R}^{KI}, 
\end{align}
and
\begin{align}
\mathbf{r}_{\textrm{DL}} &\triangleq \bracket{\mathit{R}_{\textrm{d},1}\bracket{0},\cdots,\mathit{R}_{\textrm{d},J}\bracket{K-1}}^\top\in\mathbb{R}^{KJ},
\end{align} 
which lead to the Lagrangian associated with \eqref{WMMSE2} as
\begin{flalign}
\label{Lagrange}
&\mathcal{L}\paren{\braces{\mathbf{P}},\mathbf{A},\boldsymbol{\lambda},\boldsymbol{\mu}}=\Xi_{\textrm{wmmse}}+\boldsymbol{\lambda}_{\textrm{DL}}^\top\paren{\mathbf{p}_{\textrm{DL}}-P_\textrm{B}\mathbf{1}}+\boldsymbol{\lambda}_{\textrm{UL}}^\top\paren{\mathbf{p}_{\textrm{UL}}-P_\textrm{U}\mathbf{1}}\nonumber\\
&-\boldsymbol{\mu}_{\textrm{DL}}^\top\paren{\mathbf{r}_{\textrm{DL}}-\mathit{R}_{\textrm{DL}}\mathbf{1}}-\boldsymbol{\mu}_{\textrm{UL}}^\top\paren{\mathbf{r}_{\textrm{UL}}-\mathit{R}_{\textrm{UL}}\mathbf{1}},
\end{flalign}\normalsize
where $\boldsymbol{\lambda}=\bracket{\boldsymbol{\lambda}^\top_{\text{DL}},\boldsymbol{\lambda}^\top_{\text{UL}}}^\top$ and $\boldsymbol{\mu}=\bracket{\boldsymbol{\mu}^\top_{\text{DL}},\boldsymbol{\mu}^\top_{\text{UL}}}^\top$. The Lagrange dual function of $L\paren{\cdot}$ is defined as
$D\paren{\boldsymbol{\lambda},\boldsymbol{\mu}}=\underset{\braces{\mathbf{P}},\mathbf{A}}\inf \mathcal{L}\paren{\braces{\mathbf{P}},\mathbf{A},\boldsymbol{\lambda},\boldsymbol{\mu}}$. 
With these definitions, we state the following theorem to solve $\paren{\ref{WMMSE2}}$.
\begin{prop}
\label{theorem: dual}
Linearize, using Taylor series approximation, $R_{\textrm{u},i}\bracket{k}$ and $R_{\textrm{d},j}\bracket{k}$, $\forall\braces{i,j,k}$ in the QoS constraints of problem $\eqref{WMMSE2}$. Then, the solution of the resulting problem is equivalent to that of its Lagrange dual problem 
\begin{equation}
\label{dualproblem}
\text{maximize} \quad D\paren{\boldsymbol{\lambda},\boldsymbol{\mu}}\text{  subject to}\quad  \boldsymbol{\lambda}  \succeq \mathbf{0}, \boldsymbol{\mu} \succeq \mathbf{0}.    
\end{equation}
\end{prop}
\begin{proof}
Solving $\paren{\ref{dualproblem}}$ yields the lower bounds of $\paren{\ref{WMMSE2}}$. The difference between the lower bound and the actual optimal value is the optimal duality gap. To equivalently obtain $\braces{\mathbf{P}^\star}$ and $\mathbf{A}^\prime$ with $\paren{\ref{dualproblem}}$, strong duality ought to hold for the primal problem $\paren{\ref{WMMSE1}}$, i.e., the optimal duality gap should be zero. However, the QoS constraints \eqref{DLrate} are non-concave leading to a non-zero duality gap. To bypass this problem, we apply the Taylor series to obtain linear approximations of $\mathit{R}_{\textrm{u},i}\bracket{k}$ and $\mathit{R}_{\textrm{d},j}\bracket{k}$. The first-order Taylor series expansion of a real-valued function with complex-valued matrix arguments $f\paren{\mathbf{X},\mathbf{X}^\ast}: \mathbb{C}^{N\times \mathrm{Q}}\times\mathbb{C}^{N\times \mathrm{Q}}\rightarrow\mathbb{R}$ around $\mathbf{X}_{\mathrm{0}}$ yields \cite{MSE_FD} 
\begin{flalign}
\label{eq: Taylor}
f\paren{\mathbf{X},\mathbf{X}^\ast} &= f\paren{\mathbf{X}_{\mathrm{0}},\mathbf{X}^\ast_{\mathrm{0}}}+\vect^\top\paren{\frac{\partial }{\partial \mathbf{X}_{\mathrm{0}}}f\paren{\mathbf{X}_{\mathrm{0}},\mathbf{X}^\ast_{\mathrm{0}}}}\vect\paren{\mathbf{X}-\mathbf{X}_{\mathrm{0}}}\nonumber\\
&+\vect^\top\paren{\frac{\partial }{\partial \mathbf{X}^\ast_{\mathrm{0}}}f\paren{\mathbf{X}_{\mathrm{0}},\mathbf{X}^\ast_{\mathrm{0}}}}\vect\paren{\mathbf{X}^\ast-\mathbf{X}^\ast_{\mathrm{0}}}.
\end{flalign}\normalsize
based on their associated Taylor series expansions in the initial approximations $\widetilde{\mathbf{P}}_{\textrm{u},i}\bracket{k}$, $\widetilde{\mathbf{P}}_{\textrm{d},j}\bracket{k}$, and $\widetilde{\mathbf{a}}\bracket{k}$. Denoting $\widetilde{\mathbf{P}}_{\textrm{u},i}\bracket{k}$ as an initial approximation of $\PiB$, the Taylor series expansions of $\mathit{R}_{\textrm{u},q}\bracket{k}$ at $\widetilde{\mathbf{P}}_{\textrm{u},i}\bracket{k}$ are 
\begin{align}
\mathit{R}_{\textrm{u},q}\bracket{k}\paren{\PiB}&\approx \mathit{R}_{\textrm{u},q}\bracket{k}\paren{\widetilde{\mathbf{P}}_{\textrm{u},i}\bracket{k}}\nonumber\\
&+\vect\braces{\frac{\partial \paren{\mathit{R}_{\textrm{u},q}\bracket{k}\paren{\widetilde{\mathbf{P}}_{\textrm{u},i}\bracket{k},\widetilde{\mathbf{P}}^\ast_{\textrm{u},i}\bracket{k}}}}{\partial \widetilde{\mathbf{P}}_{\textrm{u},i}\bracket{k}}}^\top\vect\braces{\PiB-\widetilde{\mathbf{P}}_{\textrm{u},i}\bracket{k}}\nonumber\\
&+\vect\braces{\frac{\partial \paren{\mathit{R}_{\textrm{u},q}\bracket{k}\paren{\widetilde{\mathbf{P}}_{\textrm{u},i}\bracket{k},\widetilde{\mathbf{P}}^\ast_{\textrm{u},i}\bracket{k}}}}{\partial \widetilde{\mathbf{P}}^\ast_{\textrm{u},i}\bracket{k}}}^\top \vect\braces{\mathbf{P}^\ast_{\textrm{u},i}\bracket{k}-\widetilde{\mathbf{P}}^\ast_{\textrm{u},i}\bracket{k}}.\nonumber
\end{align}
Likewise, by defining $\widetilde{\mathbf{P}}_{\textrm{d},j}\bracket{k}$ and $\widetilde{\mathbf{a}}\bracket{k}$ as initial approximations of $\PBj$, and $\mathbf{a}\bracket{k}$, we find the linear approximations of $\mathit{R}_{\textrm{u},q}\bracket{k}\paren{\PBj}$, $\mathit{R}_{\textrm{u},q}\bracket{k}\paren{\mathbf{a}\bracket{k}}$, $\mathit{R}_{\textrm{d},j}\bracket{k}\paren{\PiB}$, $\mathit{R}_{\textrm{d},m}\bracket{k}\paren{\PBj}$, and $\mathit{R}_{\textrm{d},j}\bracket{k}\paren{\mathbf{a}\bracket{k}}$. As $\Xi_{\text{wmmse}}$ is multi-convex, namely that $\Xi_{\text{wmmse}}$ is not jointly convex in $\PiB$, $\PBj$, and $\mathbf{a}\bracket{k}$ but convex in each individual variable provided the rest of the variables are fixed \cite{FD_WMMSE,BCDconvergence}, we thus have a fully convex approximation of \eqref{WMMSE2} in \eqref{dualproblem} and reduce the optimal duality gap to zero \cite{Lui2006subg}. This concludes the proof.
\end{proof}
Following Appendix~\ref{app:grad}, the gradients of Lagrangian $\mathcal{L}\paren{\cdot}$ w.r.t. $\PiB$, $\PBj$, and $\mathbf{a}\bracket{k}$ are, respectively,  
\begin{flalign}
\nabla_{\PiB}\mathcal{L}&=\nabla_{\PiB}\Xi_{\text{UL}}+\nabla_{\PiB}\Xi_{\text{DL}}+\nabla_{\PiB}\Xi_{\text{r}}\nonumber\\
&+\lambda_{\textrm{u},i}\bracket{k}\PiB-\sum_{g=1}^{J}\mu_{\textrm{d},g}\bracket{k}\nabla_{\PiB}\mathit{R}_{\textrm{d},g}\bracket{k}-\sum_{q=1}^{\mathit{I}}\mu_{\textrm{u},q}\bracket{k}\nabla_{\PiB}\mathit{R}_{\textrm{u},q}\bracket{k},\nonumber\\
\nabla_{\PBj}\mathcal{L}&=\nabla_{\PBj}\Xi_{\text{UL}}+\nabla_{\PBj}\Xi_{\text{DL}}+\nabla_{\PBj}\Xi_{\text{r}}\nonumber\\
&+\lambda_{\textrm{d}}\bracket{k}\PBj-\sum_{g=1}^{J}\mu_{\textrm{d},g}\bracket{k}\nabla_{\PBj}\mathit{R}_{\textrm{d},g}\bracket{k}-\sum_{q=1}^{\mathit{I}}\mu_{\textrm{u},q}\bracket{k}\nabla_{\PBj}\mathit{R}_{\textrm{u},q}\bracket{k},\nonumber\\
\text{and }\nabla_{\mathbf{a}\bracket{k}}\mathcal{L}&=\nabla_{\mathbf{a}\bracket{k}}\Xi_{\text{UL}}+\nabla_{\mathbf{a}\bracket{k}}\Xi_{\text{DL}}+\nabla_{\mathbf{a}\bracket{k}}\Xi  _{\text{r}}-\sum_{j=1}^{\mathit{J}}\mu_{\textrm{d},j}\bracket{k}\nabla_{\mathbf{a}\bracket{k}}R_{\textrm{d},j}\bracket{k}-\sum_{i=1}^{\mathit{I}}\mu_{\textrm{u},i}\bracket{k}\nabla_{\mathbf{a}\bracket{k}}R_{\textrm{u},i}\bracket{k}.\nonumber
\end{flalign}\normalsize
We obtain the analytical solutions to $\PiB$, $\PBj$, and $\mathbf{a}\bracket{k}$ by solving equations $\triangledown_{\mathbf{P}_{\textrm{u},i}\bracket{k}}L=0$, $\triangledown_{\mathbf{P}_{\textrm{d},j}\bracket{k}}L=0$, and $\triangledown_{\mathbf{a}\bracket{k}}L=0$, respectively.
Using the gradient expressions in Appendix~\ref{app:grad}, some  algebra  yields three generalized Sylvester equations
\begin{align}
\mathbf{A}_{\textrm{u},i}\bracket{k}\mathbf{P}_{\textrm{u},i}\bracket{k}+\mathbf{F}_{\textrm{u},i}\bracket{k}\mathbf{P}_{\textrm{u},i}\bracket{k}\mathbf{B}_{\textrm{u,}i}\bracket{k}=\mathbf{C}_{\textrm{u,}i}\bracket{k},\\
\mathbf{A}_{\textrm{d},j}\bracket{k}\mathbf{P}_{\textrm{d},j}\bracket{k}+\mathbf{F}_{\textrm{Bt}}\bracket{k} \mathbf{P}_{\textrm{d},j}\bracket{k}\mathbf{B}_{\textrm{Bt},j}\bracket{k}+\mathbf{F}_{\textrm{Bm}}\bracket{k}\mathbf{P}_{\textrm{d},j}\bracket{k}\mathbf{B}_{\textrm{Bm},j}\bracket{k}
&=\mathbf{C}_{\textrm{d,}j}\bracket{k},\\
\textrm{and }
\mathbf{A}_{\textrm{r}}\bracket{k}\mathbf{a}\bracket{k}+\mathbf{F}_{\textrm{r}}\bracket{k}\mathbf{a}\bracket{k}&=\mathbf{c}_{\textrm{r}}\bracket{k},
\end{align}
where 
\begin{align}
\mathbf{B}_{\textrm{u,}i}\bracket{k} &=\mathbf{d}_{\mathrm{u},i}\bracket{k,n_\target-n_{\textrm{u}}} \mathbf{d}^\dagger_{\textrm{u},i}\bracket{k,n_\target-n_{\textrm{u}}},\nonumber\\ 
\mathbf{F}_{\textrm{u},i}\bracket{k}&= 2\sum_{n\rr=1}^{N\rr}\mathrm{Re}\paren{\mathrm{\xi}_{\mathrm{r},n\rr}\paren{k,k}\boldsymbol{\Sigma}^{\paren{k,k}}_{i,n\rr}},\nonumber\\ 
\mathbf{B}_{\textrm{Bt},j}\bracket{k}&=\mathbf{d}_{\mathrm{d},j}\bracket{k,0}\mathbf{d}^\dagger_{\textrm{d},j}\bracket{k,0},\nonumber\\ \mathbf{B}_{\textrm{Bm},j}\bracket{k}&=\mathbf{d}_{\textrm{d},j}\bracket{k,n_\target-n_{\textrm{Bm}}} \mathbf{d}^\dagger_{\textrm{d},j}\bracket{k,n_\target-n_{\textrm{Bm}}},\nonumber\\ 
\mathbf{F}_{\textrm{Bt}}\bracket{k}&=2\sum_{n\rr=1}^{N\rr}\mathrm{Re}\paren{\mathrm{\xi}_{\mathrm{r},n\rr}\paren{k,k}\boldsymbol{\Sigma}^{\paren{k,k}}_{\mathrm{Bt},n\rr}},\nonumber\\
\mathbf{F}_{\textrm{Bm},j}\bracket{k}&=2\sum_{n\rr=1}^{N\rr}\mathrm{Re}\paren{\mathrm{\xi}_{\mathrm{r},n\rr}\paren{k,k}\boldsymbol{\Sigma}^{\paren{k,k}}_{\mathrm{Bm},n\rr}},\nonumber\\  
\mathbf{F}_{\textrm{r}}\bracket{k}&=2\sum_{n\rr=1}^{N\rr}\bracket{\mathrm{Re}\paren{\mathrm{\xi}_{\mathrm{r},n\rr}\paren{k,k}\boldsymbol{\Sigma}^{\paren{k,k}}_{\mathrm{rt},n\rr}}+\mathrm{Re}\paren{\mathrm{\xi}_{\mathrm{r},n\rr}\paren{k,k}\boldsymbol{\Sigma}_{\textrm{c},n\rr}}},\nonumber\\
\mathbf{F}_{\textrm{r}}\bracket{k}&=2\sum_{n\rr=1}^{N\rr}\bracket{\mathrm{Re}\paren{\mathrm{\xi}_{\mathrm{r},n\rr}\paren{k,k}\boldsymbol{\Sigma}^{\paren{k,k}}_{\mathrm{rt},n\rr}}+\mathrm{Re}\paren{\mathrm{\xi}_{\mathrm{r},n\rr}\paren{k,k}\boldsymbol{\Sigma}_{\textrm{c},n\rr}}},\nonumber\\
\mathbf{A}_{\textrm{u,}i}\bracket{k}&=2\HiBH\boldsymbol{\xi}_{\textrm{UL}}\bracket{k}\mathbf{H}_{i,\B}+2\sum_{g=1}^\mathit{J}\mathbf{H}^\dagger_{i,g}\boldsymbol{\xi}_{\textrm{d},g}\bracket{k}\mathbf{H}_{i,g}+\lambda_{\textrm{u},i}\bracket{k}\mathbf{I},\nonumber\\
\mathbf{C}_{\textrm{u,}i}\bracket{k}&= \sum_{q=1}^{\mathit{I}}\mu_{\textrm{u},q}\bracket{k}\nabla_{\PiB}\mathit{R}_{\textrm{u},q}\bracket{k}+\sum_{g=1}^{J}\mu_{\textrm{d},g}\bracket{k}\nabla_{\PiB}\mathit{R}_{\textrm{d},g}\bracket{k}\nonumber\\
&-\sum_{n\rr=1}^{N\rr}\sum_{m\neq k}^{\mathrm{\mathit{K}}}2\mathrm{Re}\paren{\mathrm{\xi}_{\mathrm{r},n\rr}\paren{k,m}\boldsymbol{\Sigma}^{\paren{m,k}}_{i,n\rr}}\mathbf{P}_{\textrm{u},i}\bracket{m}\mathbf{d}_{\textrm{u},i}\bracket{m,n_\target-n_{\textrm{u}}}\mathbf{d}^\dagger_{\textrm{u},i}\bracket{k,n_{\textrm{t}}-n_{\textrm{u}}}+2\alpha^{\textrm{u}}_i\HiBH\UiBH\WiB,\nonumber
\end{align}
\begin{flalign}
\mathbf{A}_{\textrm{d,}j}\bracket{k}&=2\HBBH\boldsymbol{\xi}_{\textrm{UL}}\bracket{k}\HBB+2\sum_{g=1}^{\mathit{J}}\mathbf{H}^\dagger_{\B,g}\boldsymbol{\xi}_{\textrm{d},g}\bracket{k}\mathbf{H}_{\B,g}+\lambda_\textrm{d}\bracket{k}\mathbf{I},\nonumber
\end{flalign}
\begin{flalign}
\mathbf{C}_{\textrm{d,}j}\bracket{k}&= \sum_{q=1}^{\mathit{I}}\mu_{\textrm{u},q}\bracket{k}\nabla_{\PBj}\mathit{R}_{\textrm{u},q}\bracket{k}+\sum_{g=1}^{\mathit{J}}\mu_{\textrm{d},g}\bracket{k}\nabla_{\PBj}\mathit{R}_{\textrm{d},g}\bracket{k}-\nonumber\\
&\sum_{n\rr=1}^{N\rr}\left\lbrace\sum_{m\neq k}2\mathrm{Re}\paren{\mathrm{\xi}_{\mathrm{r},n\rr}\paren{m,k}\boldsymbol{\Sigma}^{\paren{k,m}}_{\mathrm{Bt},n\rr}}\sum_{g=1 }^J\mathbf{P}_{\textrm{d},g}\bracket{m}\mathbf{d}_{\mathrm{d
},g}\bracket{m,0}\mathbf{d}^\dagger_{\textrm{d},j}\bracket{k,0}\right.\nonumber\\
&+\sum_{m\neq k}2\mathrm{Re}\paren{\mathrm{\xi}_{\mathrm{r},n\rr}\paren{m,k}\boldsymbol{\Sigma}^{\paren{k,m}}_{\mathrm{Bm},n\rr}}\sum_{g=1 }^J\mathbf{P}_{\textrm{d},g}\bracket{m}\mathbf{d}_{\mathrm{d},g}\bracket{m,n_\target-n_{\textrm{Bm}}}\mathbf{d}^\dagger_{\textrm{d},j}\bracket{k,n_{\textrm{t}}-n_{\textrm{Bm}}}\nonumber\\
&+\sum_{g\neq j}2\mathrm{Re}\paren{\mathrm{\xi}_{\mathrm{r},n\rr}\paren{k,k}\boldsymbol{\Sigma}^{\paren{k,k}}_{\mathrm{Bt},n\rr}}\mathbf{P}_{\textrm{d},g}\bracket{k}\mathbf{d}_{\mathrm{d},g}\bracket{k,0}\mathbf{d}^\dagger_{\textrm{d},j}\bracket{k,0}\nonumber\\
&+\sum_{g\neq j}2\mathrm{Re}\paren{\mathrm{\xi}_{\mathrm{r},n\rr}\paren{k,k}\boldsymbol{\Sigma}^{\paren{k,k}}_{\mathrm{Bm},n\rr}}\mathbf{P}_{\textrm{d},g}\bracket{k}\mathbf{d}_{\mathrm{d},g}\bracket{k,n_\target-n_{\textrm{Bm}}}\mathbf{d}^\dagger_{\textrm{d},j}\bracket{k,n_{\textrm{t}}-n_{\textrm{Bm}}}\nonumber\\
&\left.-2\alpha^{\textrm{r}}_{n\rr}\sum_{m=1}^{K}\boldsymbol{\Sigma}^{\paren{m,k}}_{\mathrm{Bt},n\rr}\mathbf{J}^\top_{\mathrm{B}}\mathbf{J}^\top_{\mathrm{h}}\bracket{m}\Wrnr\urk\mathbf{d}^\dagger_{\textrm{d},j}\bracket{k,0}\right\rbrace+2\alpha^{\textrm{d}}_j\HBjH\UBjH\WBj,\nonumber\\
\mathbf{A}_{\textrm{r}}\bracket{k}&=2\HrBH\boldsymbol{\xi}_{\textrm{UL}}\bracket{k}\HrB+2\sum_{j=1}^{\mathit{J}}\HrjH\boldsymbol{\xi}_{\textrm{d},j}\bracket{k}\Hrj,\nonumber\\
\text{and }\mathbf{c}_{\textrm{r}}\bracket{k}&= \sum_{i=1}^{\mathit{I}}\mu_{\textrm{u},i}\bracket{k}\nabla_{\mathbf{a}\bracket{k}}\mathit{R}_{\textrm{u},i}\bracket{k}+\sum_{j=1}^{\mathit{J}}\mu_{\textrm{d},j}\bracket{k}\nabla_{\mathbf{a}\bracket{k}}\mathit{R}_{\textrm{d},j}\bracket{k}+2\sum_{n\rr=1}^{N\rr}\alpha_{\textrm{r},n\rr}\sum_{m=1}^{K}\mathrm{Re}\paren{\boldsymbol{\Sigma}^{\paren{m,k}}_{\mathrm{rt},n\rr}}\mathbf{J}^\top\rr\mathbf{J}^\top_{\mathrm{h}}\bracket{m}\Wrnr\urk\nonumber\\
&-2\sum_{n\rr=1}^{N\rr}\sum_{m\neq k}\bracket{\mathrm{Re}\paren{\mathrm{\xi}_{\mathrm{r},n\rr}\paren{k,m}\boldsymbol{\Sigma}^{\paren{m,k}}_{\mathrm{rt},n\rr}}+\mathrm{Re}\paren{\mathrm{\xi}_{\mathrm{r},n\rr}\paren{k,m}\boldsymbol{\Sigma}_{\textrm{c},n\rr}}}\mathbf{a}\bracket{m}.\nonumber
\end{flalign}
\normalsize
Solving these Sylvester equations yields \cite{MSE_FD}
\begin{subequations}
\begin{align}
\vect\paren{\mathbf{P}^\star_{\textrm{u},i}\bracket{k}}&=\bracket{\mathbf{I}_{D^{\textrm{u}}_{i}}\otimes\mathbf{A}_{\textrm{u},i}\bracket{k}+\mathbf{B}^\top_{\textrm{u},i}\bracket{k}\otimes\mathbf{F}_{\textrm{u},i}\bracket{k}}^{-1}\vect\paren{\mathbf{C}_{\textrm{u,}i}\bracket{k}},\label{PiB}\\
\vect\paren{\mathbf{P}^\star_{\textrm{d},j}\bracket{k}}&=\left(\mathbf{I}\otimes\mathbf{A}_{\mathrm{d},j}\bracket{k}+\mathbf{B}^\top_{\textrm{Bt},j}\bracket{k}\otimes\mathbf{F}_{\textrm{Bt},j}\bracket{k}+\mathbf{B}^\top_{\textrm{Bm},j}\bracket{k}\otimes\mathbf{F}_{\textrm{Bm},j}\bracket{k}\right)^{-1}\vect\paren{\mathbf{C}_{\textrm{d,}j}\bracket{k}},\label{PBj}
\\
\textrm{and\;}\mathbf{a}^\prime\bracket{k}&=\bracket{1\otimes\mathbf{A}_{\textrm{r}}\bracket{k}+1\otimes\mathbf{F}_{\textrm{r}}\bracket{k}}^{-1}\mathbf{c}_{\textrm{r}}\bracket{k},\label{ak}
\end{align}
\end{subequations}
\normalsize
for all $i,j,k$. In order to determine $\mathbf{P}^\star_{\textrm{u},i}\bracket{k}$, $\mathbf{P}^\star_{\textrm{d},j}\bracket{k}$, and $\mathbf{a}^\prime\bracket{k}$, we need to find the optimal $\boldsymbol{\lambda}$ and $\boldsymbol{\mu}$ denoted by 
$\boldsymbol{\lambda}^\star$ and $\boldsymbol{\mu}^\star$. 
\color{black}
\subsubsection{Sub-gradient method for precoders}As $D\paren{\lambda,\mu}$ is not always differentiable \cite{Lui2006subg} and a simple method to update $\boldsymbol{\lambda}$ and $\boldsymbol{\mu}$ is needed, we resort to the subgradient to determine search directions for $\boldsymbol{\lambda}$ and $\boldsymbol{\mu}$, and employ the projected subgradient method  to solve \eqref{dualproblem}  sequentially. In the $\ith{t}$ iteration we have, 
\begin{subequations}
\begin{align}
&\lambda^{\paren{t+1}}_{\textrm{u},i}\bracket{k} =\bracket{\lambda^{\paren{t}}_{\textrm{u},i}\bracket{k}+\mathrm{\beta}^{\paren{t}}_{\textrm{u},i}\bracket{k}\paren{P^{\paren{t}}_{\textrm{u},i}\bracket{k}-P_\textrm{U}}}^+,\label{lambda_UL}
\end{align}
\begin{align}
&\lambda^{\paren{t+1}}_{\textrm{d}}\bracket{k} =\bracket{\lambda^{\paren{t}}_{\textrm{d}}\bracket{k}+\beta^{\paren{t}}_{\textrm{d}}\bracket{k}\paren{P^{\paren{t}}_{\textrm{B}}\bracket{k}-P_\textrm{B}}}^+,\label{lambda_DL}
\end{align}
\begin{align}
&\mu^{\paren{t+1}}_{\textrm{u},i}\bracket{k} = \bracket{\mu^{\paren{t}}_{\textrm{u},i}\bracket{k}+\varepsilon^{\paren{t}}_{\textrm{u},i}\bracket{k}\paren{\mathit{R}_{\textrm{UL}}-\mathit{R}^{\paren{t}}_{\textrm{u},i}\bracket{k}}}^+,	\label{mu_UL}
\end{align}
\begin{align}
\text{and }&\mu^{\paren{t+1}}_{\textrm{d},j}\bracket{k} = \bracket{\mu^{\paren{t}}_{\textrm{d},j}\bracket{k}+\varepsilon^{\paren{t}}_{\textrm{d},j}\bracket{k}\paren{\mathit{R}_\textrm{DL}-\mathit{R}^{\paren{t}}_{\textrm{d},j}\bracket{k}}}^+,\label{mu_DL} 
\end{align}
\end{subequations}\normalsize
where  $\beta^{\paren{t}}_{\textrm{u},i}\bracket{k}$, $\beta^{\paren{t}}_{\textrm{d}}\bracket{k}$, $\varepsilon^{\paren{t}}_{\textrm{u},i}\bracket{k}$, and $\varepsilon^{\paren{t}}_{\textrm{d},j}\bracket{k}$ denote the step sizes of the $\ith{t}$ iteration for $\lambda^{\paren{t}}_{\textrm{u},i}\bracket{k}$, $\lambda^{\paren{t}}_{\textrm{d}}\bracket{k}$, $\mu^{\paren{t}}_{\textrm{u},i}\bracket{k}$, and $\mu^{\paren{t}}_{\textrm{d},j}\bracket{k}$, respectively, 
$P^{\paren{t}}_{\textrm{u},i}\bracket{k}$, $P^{\paren{t}}_{\textrm{B}}\bracket{k}$,
$\mathbf{P}^{\paren{t}}_{\textrm{u},i}\bracket{k}$, $\mathbf{P}^{\paren{t}}_{\textrm{d},j}\bracket{k}$, $\mathit{R}^{\paren{t}}_{\textrm{u},i}\bracket{k}$, $\mathit{R}^{\paren{t}}_{\textrm{d},j}\bracket{k}$, and $\Xi^{\paren{t}}_{\textrm{wmmse}}$ the $\ith{t}$ iterates of $P_{\textrm{u},i}\bracket{k}$, $P_{\textrm{B}}\bracket{k}$,  $\PiB$, $\PBj$, $\mathit{R}_{\textrm{u},i}\bracket{k}$, $\mathit{R}_{\textrm{d},j}\bracket{k}$, and $\Xi_{\textrm{wmmse}}$, respectively \cite{Lui2006subg}. Note that $\mathbf{P}^{\paren{t}}_{\textrm{u},i}\bracket{k}$ and $\mathbf{P}^{\paren{t}}_{\textrm{d},j}\bracket{k}$ are obtained by replacing $\lambda_{\textrm{u},i}\bracket{k}$ and $\mu_{\textrm{u},i}\bracket{k}$ with $\lambda^{\paren{t}}_{\textrm{u},i}\bracket{k}$ and $\mu^{\paren{t}}_{\textrm{u},i}\bracket{k}$ in $\paren{\mathrm{\ref{PiB}}}$, and $\lambda_{\textrm{d}}\bracket{k}$ and $\mu_{\textrm{d},j}\bracket{k}$ with $\lambda^{\paren{t}}_{\textrm{d}}\bracket{k}$ and $\mu^{\paren{t}}_{\textrm{d},j}\bracket{k}$ in $\paren{\mathrm{\ref{PBj}}}$, respectively.

\color{black}There are various options to choose $\beta^{\paren{t}}_{\textrm{u},i}\bracket{k}$, $\beta^{\paren{t}}_{\textrm{d}}\bracket{k}$, $\varepsilon^{\paren{t}}_{\textrm{u},i}\bracket{k}$, and $\varepsilon^{\paren{t}}_{\textrm{d},j}\bracket{k}$ to ensure that the subgradient updates in \eqref{lambda_UL}-\eqref{mu_DL} converge to optimal values $\boldsymbol{\lambda}^\star$ and $\boldsymbol{\mu}^\star$ of $\boldsymbol{\lambda}$ and $\boldsymbol{\mu}$, respectively. Broadly, two rules are used. The first determines the step size before executing the algorithm. 
This includes fixed and progressively diminishing step sizes; the latter should be square summable (but not necessarily summable). The second rule computes the steps using, for example, the Polyak's procedure \cite{Boydsubgradientnote}. As discussed in \cite{Boydsubgradientnote} and \cite{Lui2006subg}, the convergence rate of subgradient is dependent on step sizes and initialization points. It was mentioned in \cite{Boydsubgradientnote} that the Polyak’s rule achieves faster convergence rate than other rules because it utilizes the optimal function value of the current iteration while computing the step size. Applying the Polyak's rule to $\beta^{\paren{t}}_{\textrm{u},i}\bracket{k}$, $\beta^{\paren{t}}_{\textrm{d}}\bracket{k}$,  $\varepsilon^{\paren{t}}_{\textrm{u},i}\bracket{k}$, and $\varepsilon^{\paren{t}}_{\textrm{d},j}\bracket{k}$ produces
\begin{subequations}
\begin{align}
\label{eq: beta_ui_k_t}
\beta^{\paren{t}}_{\textrm{u},i}\bracket{k}=\sfrac{\paren{\Xi^{\paren{t}}_{\textrm{wmmse}}-\Xi^{\textrm{min}}_{\textrm{wmmse}}+0.1^t}}{\left|P^{\paren{t}}_{\textrm{u},i}\bracket{k}-P_\textrm{U}\right|^2},
\end{align}
\begin{align}
\label{eq: beta_d_k_t}
\beta^{\paren{t}}_{\textrm{d}}\bracket{k}=\sfrac{\paren{\Xi^{\paren{t}}_{\textrm{wmmse}}-\Xi^{\textrm{min}}_{\textrm{wmmse}}+0.1^t}}{\left|P^{\paren{t}}_{\textrm{B}}\bracket{k}-P_\textrm{B}\right|^2},
\end{align}
\begin{flalign}
\label{eq: epsilon_ui_k_t}
\varepsilon^{\paren{t}}_{\textrm{u},i}\bracket{k} = \sfrac{\paren{\Xi^{\paren{t}}_{\textrm{wmmse}}-\Xi^{\textrm{min}}_{\textrm{wmmse}}+0.1^t}}{\left|R_{\textrm{UL}}-R^{\paren{t}}_{\textrm{u},i}\bracket{k}\right|^2},
\end{flalign}
\begin{flalign}
\label{eq: epsilon_dj_k_t}
\textrm{and }\varepsilon^{\paren{t}}_{\textrm{d},j}\bracket{k} = \sfrac{\paren{\Xi^{\paren{t}}_{\textrm{wmmse}}-\Xi^{\textrm{min}}_{\textrm{wmmse}}+0.1^t}}{\left|R_{\textrm{DL}}-R^{\paren{t}}_{\textrm{u},i}\bracket{k}\right|^2}.
\end{flalign}
\end{subequations}\normalsize 
Denote $\ell_{\textrm{max}}$, $\iota_{\textrm{max}}$, $t_{\textrm{u,max}}$($t_{\textrm{d,max}}$) denote the maximum iterations for the BCD-AP MRMC, WMMSE-MRMC, and the subgradient algorithms, respectively; $(\cdot)^{\paren{\ell,\iota,t}}$ as the iterate of a variable at the $\ith{\ell}$, $\ith{\iota}$, and $\ith{t}$  iterations of BCD-AP MRMC, WWMSE-MRMC, and subgradient algorithms; and $\braces{\mathbf{P}^{\paren{\ell,\iota,0}}}\triangleq\braces{\mathbf{P}^{\paren{\ell,\iota,0}}_{\textrm{u},i}\bracket{k},\mathbf{P}^{\paren{\ell,\iota,0}}_{\textrm{d},j}\bracket{k},\forall \braces{i,j,k}}$. The sub-gradient method is not descent-based and $\Xi^{\paren{t}}_{\textrm{wmmse}}$ may increase at certain iterations\cite{Lui2006subg}. Therefore, $\lambda^{\star}_{\textrm{u},i}\bracket{k}$, $\mu^\star_{\textrm{u},i}\bracket{k}$, $\lambda^{\star}_{\textrm{d}}\bracket{k}$, and $\mu^\star_{\textrm{d},j}\bracket{k}$ are employed to keep track of those values of $\lambda^{\paren{t}}_{\textrm{u},i}\bracket{k}$, $\mu^{\paren{t}}_{\textrm{u},i}\bracket{k}$, $\lambda^{\paren{t}}_{\textrm{d}}\bracket{k}$, and $\mu^{\paren{t}}_{\textrm{d},j}\bracket{k}$ that yield the minimum $\Xi^{\paren{t}}_{\textrm{wmmse}}$ in the current iteration, i.e.  $\Xi^{\textrm{min}}_{\textrm{wmmse}}$. 
Algorithm $\ref{ULalgorithm}$ summarizes the steps of the sub-gradient-based procedures to find $\lambda_{\textrm{u},i}\bracket{k}$, $\mu_{\textrm{u},i}\bracket{k}$ and corresponding $\mathbf{P}^{\paren{\mathrm{\ell,\iota}}}_{\textrm{u},i}\bracket{k}$. We set $\widetilde{\mathbf{P}}_{\textrm{u},i}\bracket{k}$ as the optimal estimates of the previous iteration (step \ref{step: tilde_P_ui_k}). We use the same procedure for DL UE to find  $\lambda_{\textrm{d}}\bracket{k}$, $\mu_{\textrm{d},j}\bracket{k}$ and $\mathbf{P}^{\paren{\mathrm{\ell,\iota}}}_{\textrm{d},j}\bracket{k}$. \color{black} 
\begin{algorithm}[ht!]
\caption{Subgradient approach to solve $\paren{\ref{dualproblem}}$ for UL UE}
\label{ULalgorithm}
\begin{algorithmic}[1]
\Statex \textbf{Input: } $\braces{\mathbf{P}^{\paren{\mathrm{\ell,\iota}}}}$,  $\mathbf{A}^{\paren{\mathrm{\ell,\iota}}}$, $\braces{\mathbf{U}^{\paren{\mathrm{\ell}}}}$, $t_{\textrm{u,max}}$
\Statex \textbf{Output:} $\mathbf{P}^{\paren{\mathrm{\ell,\iota}}}_{\textrm{u},i}\bracket{k}$, $\mu^\star_{\textrm{u},i}\bracket{k}$
\State Initialize $\lambda^{\paren{\mathrm{0}}}_{\textrm{u},i}\bracket{k}=1$, $\mu^{\paren{\mathrm{0}}}_{\textrm{u},i}\bracket{k}=1$, $\braces{\mathbf{P}^{\paren{\ell,\iota,0}}}=\braces{\mathbf{P}^{\paren{\mathrm{\ell,\iota}}}}$  
\State $t\leftarrow1$, $\Xi^{\paren{\textrm{min}}}_{\textrm{wmmse}}\xleftarrow{\braces{\mathbf{P}^{\paren{\mathrm{\ell,\iota}}}},\mathbf{A}^{\paren{\mathrm{\ell,\iota}}},\braces{\mathbf{U}^{\paren{\mathrm{\ell}}}}}$ \eqref{Xi_Mses}  
\Repeat
\State update $P^{\paren{t}}_{\textrm{u},i}\bracket{k}$, $\mathit{R}^{\paren{\textrm{t}}}_{\textrm{u},i}\bracket{k}$, $\beta^{\paren{t}}_{\textrm{u},i}\bracket{k}$ $\varepsilon^{\paren{t}}_{\textrm{u},i}\bracket{k}$, 
$\lambda^{\paren{t}}_{\textrm{u},i}\bracket{k}$ , $\mu^{\paren{t}}_{\textrm{u},i}\bracket{k}$ 
\State $\mathbf{P}^{\paren{\mathrm{\ell,\iota,t}}}_{\textrm{u},i}\bracket{k}\xleftarrow{\widetilde{\mathbf{P}}_{\textrm{u},i}\bracket{k}=\mathbf{P}^{\paren{\mathrm{\ell,\iota}}}_{\textrm{u},i}\bracket{k}}$  \eqref{PiB} \label{step: tilde_P_ui_k}
\State $\Xi^{\paren{t}}_{\textrm{wmmse}}\xleftarrow{\braces{\mathbf{P}^{\paren{\ell,\iota,t}}},\mathbf{A}^{\paren{\mathrm{\ell,\iota}}},\braces{\mathbf{U}^{\paren{\mathrm{\ell}}}}}$  \eqref{Xi_Mses} 
\If{$\Xi^{\textrm{min}}_{\textrm{wmmse}}>\Xi^{\paren{t}}_{\textrm{wmmse}}$}  $\mathbf{P}^{\paren{\mathrm{\ell,\iota}}}_{\textrm{u},i}\bracket{k}=\mathbf{P}^{\paren{\mathrm{\ell,\iota,t}}}_{\textrm{u},i}\bracket{k}$, $\lambda^\star_{\textrm{u},i}\bracket{k}=\lambda^{\paren{t}}_{\textrm{u},i}\bracket{k}$,
$\mu^\star_{\textrm{u},i}\bracket{k}=\mu^{\paren{t}}_{\textrm{u},i}\bracket{k}$, and $\Xi^{\textrm{min}}_{\textrm{wmmse}}=\Xi^{\paren{t}}_{\textrm{wmmse}}$
\EndIf
\State $t\leftarrow t+1$
\Until $t>t_{\textrm{u,max}}$		
\State\Return{$\mathbf{P}^{\paren{\mathrm{\ell,\iota}}}_{\textrm{u},i}\bracket{k}$}, $\mu^\star_{\textrm{u},i}\bracket{k}$
\end{algorithmic}
\end{algorithm}\normalsize

Upon executing the subgradient algorithms for all $i$, $j$, and $k$, $\mathbf{A}^{\paren{\ell,\iota}}$ is solved by replacing  $\boldsymbol{\mu}^\star$ and $\braces{\mathbf{P}^{\paren{\ell,\iota}}}$  
with $\boldsymbol{\mu}$ and $\braces{\mathbf{P}}$ in \eqref{ak} during the $\ith{\iota}$ iteration of the WMMSE-MRMC algorithm (Algorithm~\ref{convexalgorithm}), whose maximum number of iterations is $\mathrm{\iota}_{\textrm{max}}$. The outputs of the Algorithm~\ref{convexalgorithm} constitute the $\ith{\ell}$ iterate of $\braces{\mathbf{P}^{\paren{\ell}}}$ and $\mathbf{A}^\prime$ for the Algorithm \ref{Alternating_sum}.  

So far, Algorithm~\ref{convexalgorithm} utilized perfect chanel state information (CSI). In practice, the CSI is estimated. The resulting estimation error may be modeled either as norm-bounded or stochastically \cite{FD_WMMSE}. The former is employed when quantization error is the primary source of the CSI error. However, quantization analysis is beyond the scope of this paper. Therefore, we adopt the latter by modeling the channel matrix as $\widehat{\mathbf{H}}={\mathbf{H}}+\Delta$, where and $\Delta\sim\mathcal{CN}\paren{\mathbf{0},\eta^2_{\mathrm{CSI}}\mathbf{I}}$ is the error with the variance $\eta^2_{\mathrm{CSI}}\mathbf{I}$.
\begin{algorithm}[ht!]
\caption{WMMSE-MRMC algorithm to solve $\paren{\ref{WMMSE2}}$}
\label{convexalgorithm}
\begin{algorithmic}[1]
	\Statex \textbf{Input:} $\braces{\mathbf{P}^{\paren{\mathrm{\ell}}}}$, $\mathbf{A}^{\paren{\mathrm{\ell}}}$, $\braces{\mathbf{U}^{\paren{\mathrm{\ell}}}}$, $\mathrm{\iota}_{\textrm{max}}$, $t_{\textrm{u,max}}$, and $t_{\textrm{d,max}}$
	\Statex \textbf{Output: } $\braces{\mathbf{P}^{\paren{\ell}}}$, $\mathbf{A}^{\prime}$
	\State Set $\braces{\mathbf{P}^{\paren{\mathrm{\ell,0}}}}\triangleq\braces{\mathbf{P}^{\paren{\mathrm{\ell,0}}}_{\textrm{u},i}\bracket{k},\mathbf{P}^{\paren{\mathrm{\ell,0}}}_{\textrm{d},j}\bracket{k}, \forall \braces{i,j,k}}=\braces{\mathbf{P}^{\paren{\mathrm{\ell}}}}
	$ and $\mathbf{A}^{\paren{\mathrm{\ell,0}}}\triangleq\bracket{\paren{\mathbf{a}^{\paren{\mathrm{\ell,0}}}\bracket{0}}^\top;\cdots;\paren{\mathbf{a}^{\paren{\mathrm{\ell,0}}}\bracket{\mathrm{\mathit{K}}}}^\top}=\mathbf{A}^{\paren{\mathrm{\ell}}}$
	\State Set the iteration index $\mathrm{\iota}=0$ 
	\Repeat
	\For{$k=1,\cdots,\mathrm{\mathit{K}}$}\label{stepk} 
	\For{$i=1,\cdots,\mathit{I}$, $j=1,\cdots,\mathit{J}$}
	\State $\mathbf{P}^{\paren{\mathrm{\ell,\iota+1}}}_{\textrm{u},i}\bracket{k},\mu^\star_{\textrm{u},i}\bracket{k}\xleftarrow[t_{\textrm{u,max}},\braces{\mathbf{U}^{\paren{\mathrm{\ell}}}}]{\textrm{Subgradient}}\braces{\mathbf{P}^{\paren{\mathrm{\ell,\iota}}}},\mathbf{A}^{\paren{\mathrm{\ell,\iota}}}$ 
	\State 
	$\mathbf{P}^{\paren{\mathrm{\ell,\iota+1}}}_{\textrm{d},j}\bracket{k},\mu^\star_{\textrm{d},j}\bracket{k}\xleftarrow[t_{\textrm{d,max}},\braces{\mathbf{U}^{\paren{\mathrm{\ell}}}}]{\textrm{Subgradient }}\braces{\mathbf{P}^{\paren{\mathrm{\ell,\iota}}}},\mathbf{A}^{\paren{\mathrm{\ell,\iota}}}$ 
	\EndFor
	\State 	$\mathbf{a}^{\paren{\mathrm{\ell,\iota+1}}}\bracket{k}\xleftarrow{\eqref{ak}}\braces{\mathbf{P}^{\paren{\mathrm{\ell,\iota}}}},\mathbf{A}^{\paren{\mathrm{\ell,\iota}}}$, $\braces{\mathbf{U}^{\paren{\mathrm{\ell}}}}$, $\boldsymbol{\mu}^\star$
	\EndFor
	\State $\mathbf{A}^{\paren{\mathrm{\ell,\iota+1}}}=\bracket{\paren{\mathbf{a}^{\paren{\mathrm{\ell,\iota+1}}}\bracket{0}}^\top;\cdots;\paren{\mathbf{a}^{\paren{\mathrm{\ell,\iota+1}}}\bracket{\mathrm{\mathit{K}}}}^\top}$
	\State	$\mathrm{\iota}\leftarrow \mathrm{\iota}+1$
	\Until $\mathrm{\iota}>\mathrm{\iota}_{\textrm{max}}$
	\State $\braces{\mathbf{P}^{\paren{\ell}}}\leftarrow\braces{\mathbf{P}^{\paren{\mathrm{\ell,\iota}}}}$,  $\mathbf{A}^{\prime}\leftarrow\mathbf{A}^{\paren{\mathrm{\ell,\iota}}}$
	\State \Return $\braces{\mathbf{P}^{\paren{\ell}}}$, $\mathbf{A}^{\prime}$
\end{algorithmic}\normalsize
\end{algorithm}

\subsection{Nearest vector method to find \texorpdfstring{$\mathbf{A}^\star$}{string}}
\label{subsec: PAR}
Upon obtaining $\mathbf{A}^\prime=\bracket{\mathbf{a}^\prime_{1},\cdots,\mathbf{a}^\prime_{M\rr}}$, which is the optimal solution for $\mathbf{A}\in\mathbb{A}_{\textrm{c}}$, the next step in the BCD-AP algorithm is to apply AP for projecting $\mathbf{a}^\prime_{m\rr}$ onto $\mathbb{A}_{\textrm{r}}$. 
The nearest element of $\mathbf{a}^\prime_{m\rr}$ in $\mathbb{A}_{\textrm{r}}$ for all $m\rr$ in the following problem
\begin{equation}
\label{radarmiop}
\underset{\mathbf{a}_{m\rr},\forall m\rr}{\text{minimize}}\;\lVert\mathbf{a}_{m\rr}-\mathbf{a}^\prime_{m\rr}\rVert^2_2\text{ subject to}\;  \paren{\mathrm{\ref{constraint:radarpower}}}\textrm{ and }\paren{\mathrm{\ref{constraint:radarpar}}},
\end{equation}\normalsize
yields $\mathbf{a}^\star_{m\rr}$, the $\ith{m\rr}$ column of $\mathbf{A}^\star$. This is effectively a matrix nearness problem with specified column norms and PARs. It arises in structured tight frame design problems and is solved via AP \cite{nearestvector,arXiv180203889Z}. Using ``nearest vector with low PAR" algorithm \cite{nearestvector}, we find $\mathbf{a}^{\paren{\ell}}_{m\rr}$ for all $m\rr$ recursively at the $\ith{\ell}$ iteration of the BCD-AP MRMC algorithm (see Algorithm~\ref{PARalgorithm}).
\begin{algorithm}[H]
\caption{Nearest vector method to find $\mathbf{A}^{\paren{\ell}}$}
\label{Nearness}
\label{PARalgorithm}
	\begin{algorithmic}[1]
		\Statex \textbf{Input:} $\mathbf{A}^\prime=\bracket{\mathbf{a}^{\prime}_1,\cdots,\mathbf{a}^{\prime}_{\mathit{M}\rr}}$, $P_{\textrm{r}, m\rr}$, $\gamma_{m\rr}$, $\forall m\rr$
		\Statex \textbf{Output:} $\mathbf{A}^{\paren{\ell}}=\bracket{\mathbf{a}^{\paren{\mathrm{\ell}}}_1,\cdots,\mathbf{a}^{\paren{\ell}}_{\mathit{M}\rr}}$
		\For{$m\rr=1,\cdots,\mathit{M}\rr$}
		\State Normalize $\mathbf{a}^{\prime}_{m\rr}$ to unit norm; define $\sigma_{m\rr}=\sqrt{\sfrac{P_{\textrm{r},m\rr}\gamma_{m\rr}}{\mathit{K}}}$
		\State $P\leftarrow$ number of elements in $\mathbf{a}^{\paren{\mathrm{\ell}}}_{m\rr}$ with the least magnitude
		\State $\varpi\leftarrow$ indices of the elements in  $\mathbf{a}^{\paren{\mathrm{\ell}}}_{m\rr}$ with the least magnitude
		\If{$\min\paren{\lvert\mathbf{a}^{\paren{\mathrm{\ell}}}_{m\rr}\bracket{k}\rvert}=0, \forall k\in\varpi$}{
                \begin{equation*}
			\mathbf{a}^{\paren{\mathrm{\ell}}}_{m\rr}\bracket{k} =
			\begin{cases}
			\sqrt{\frac{P_{\textrm{r}, m\rr}\paren{K-P}\sigma^2_{m\rr}}{P}}& \text{if } k\in\varpi,
			\\
			\sigma_{m\rr}e^{j\angle{\mathbf{a}^{\paren{\mathrm{\ell}}}_{m\rr}\bracket{k}}}&\text{if } k\notin\varpi
			\end{cases}
			\end{equation*}}
		\Else
		{ 
			$\rho=\sqrt{\frac{P_{\textrm{r}, m\rr}\paren{K-P}\sigma^2_{m\rr}}{\sum_{k\in\varpi}\lvert\mathbf{a}_{m\rr}\bracket{k}\rvert^2}}$
			and 
                \begin{equation*}
		\mathbf{a}^{\paren{\mathrm{\ell}}}_{m\rr}\bracket{k}=
		\begin{cases}
		\rho\mathbf{a}^{\paren{\mathrm{\ell}}}_{m\rr}\bracket{k}& \text{if } k\in\varpi,
		\\
		\sigma_{m\rr}e^{j\angle{\mathbf{a}^{\paren{\mathrm{\ell}}}_{m\rr}\bracket{k}}}&\text{if } k\notin\varpi 
		\end{cases}		
		\end{equation*}
		}
		\EndIf
		\EndFor	
		\State\Return{$\mathbf{A}^{\paren{\ell}}=\bracket{\mathbf{a}^{\paren{\mathrm{\ell}}}_{1},\cdots,\mathbf{a}^{\paren{\mathrm{\ell}}}_{\mathit{M}\rr}}$}
	\end{algorithmic}
\end{algorithm}
Once $\braces{\mathbf{P}^{\paren{\ell}}}$ and $\mathbf{A}^{\paren{\ell}}$ are known, we update $\braces{\mathbf{U}^{\paren{\ell}}}$  with WMMSE solutions from Section~\ref{subsec: MMSE section}. Algorithm $\ref{Alternating_sum}$ summarizes the BCD-AP MRMC with $\mathrm{\ell}_{\textrm{max}}$ the maximum number of iterations. 
	\begin{algorithm}[H]
	    \caption{BCD-AP MRMC algorithm 
	    }
		\label{Alternating_sum}
		\begin{algorithmic}[1]
			\Statex \textbf{Input:} $\mathbf{A}$, $\mathrm{\ell}_{\textrm{max}}$, $\mathrm{\iota}_{\textrm{max}}$, $t_{\textrm{u,max}}$ $t_{\textrm{d,max}}$
			\Statex \textbf{Output:} Optimal UL/DL precoders $\braces{\mathbf{P}^\star}$, MIMO radar code matrix $\mathbf{A}^\star$, and LRFs $\braces{\mathbf{U}^\star}$
			\State Initialize $\braces{\mathbf{P}^{\paren{\mathrm{0}}}}\triangleq
			\braces{\mathbf{P}^{\paren{\mathrm{0}}}_{\textrm{u},i}\bracket{k},\mathbf{P}^{\paren{\mathrm{0}}}_{\textrm{d},j}\bracket{k}, \forall \braces{i,j,k}}$ and $\mathbf{A}^{\paren{\mathrm{0}}}=\bracket{\mathbf{a}^{\paren{\mathrm{0}}}\bracket{0};\cdots;\mathbf{a}^{\paren{\mathrm{0}}}\bracket{\mathrm{\mathit{K}}}}$
			\State 
			$\braces{\mathbf{U}^{\paren{\mathrm{0}}}}$ $\xleftarrow{\mathrm{\paren{\ref{radarWMMSE_Rx}}},\mathrm{\paren{\ref{UL_WMMSE_Rx}}},\mathrm{\paren{\ref{DL_WMMSE_Rx}}}}$  $\braces{\mathbf{P}^{\paren{\mathrm{0}}}}$ and $\mathbf{A}^{\paren{\mathrm{0}}}$ 
			\State Set the alternating projection iteration index $\mathrm{\ell=0}$
			\Repeat \; 
			\State $\braces{\mathbf{P}^{\paren{\mathrm{\ell+1}}}}$, $\mathbf{A}^\prime$ $\xleftarrow[\mathrm{\iota}_{\textrm{max}},t_{\textrm{u,max}},t_{\textrm{d,max}}]{\textrm{Algorithm }\ref{convexalgorithm}}$ $\braces{\mathbf{P}^{\paren{\mathrm{\ell}}}}$, $\mathbf{A}^{\paren{\mathrm{\ell}}}$, $\braces{\mathbf{U}^{\paren{\mathrm{\ell}}}}$  
			\State $\mathbf{A}^{\paren{\mathrm{\ell+1}}}\xleftarrow{\textrm{Algorithm }\ref{Nearness}}\mathbf{A}^\prime$  
			\State $\braces{\mathbf{U}^{\paren{\mathrm{\ell+1}}}}$ $\xleftarrow{\mathrm{\paren{\ref{radarWMMSE_Rx}}},\mathrm{\paren{\ref{UL_WMMSE_Rx}}},\mathrm{\paren{\ref{DL_WMMSE_Rx}}}}$  $\braces{\mathbf{P}^{\paren{\mathrm{\ell+1}}}}$ and $\mathbf{A}^{\paren{\mathrm{\ell+1}}}$ 
			\State $\mathrm{\ell}\leftarrow\mathrm{\ell}+1$
			\Until $\mathrm{\ell}>\mathrm{\ell}_{\textrm{max}}$
			\State $\braces{\mathbf{P}^\star}\leftarrow\braces{\mathbf{P}^{\paren{\mathrm{\ell}}}}$, $\mathbf{A}^\star\leftarrow\mathbf{A}^{\mathrm{\ell}}$, and $\braces{\mathbf{U}^\star}\leftarrow\braces{\mathbf{U}^{\paren{\mathrm{\ell}}}}$
			\State \Return{$\braces{\mathbf{P}^\star},\mathbf{A}^\star,\braces{\mathbf{U}^\star}$}
		\end{algorithmic}
	\end{algorithm}\normalsize

\subsection{Complexity and Convergence}
\label{sec: Complexity}
The subgradient algorithms are guaranteed to converge to $\boldsymbol{\lambda}^\star$ and $\boldsymbol{\mu}^\star$ as long as $\beta^{\paren{t}}_{\textrm{u},i}\bracket{k}$, $\beta^{\paren{t}}_{\textrm{d}}\bracket{k}$, $\varepsilon^{\paren{t}}_{\textrm{u},i}\bracket{k}$, and $\varepsilon^{\paren{t}}_{\textrm{d},j}\bracket{k}$ are sufficiently small; their computational complexities are $\mathcal{O}\paren{\mathit{I}}$ and $\mathcal{O}\paren{\mathit{J}}$, respectively \cite{Lui2006subg}. The WMMSE-MRMC algorithm converges locally because its alternating procedure produces a monotonically non-increasing sequence of iterates, $\braces{\Xi^{\paren{\ell,\iota}}_{\textrm{wmmse}}}$; see Appendix C of \cite{Luo2011IterativeWMMSE} for proof. However, $\Xi_{\textrm{wmmse}}$ is not jointly convex on $\braces{\mathbf{P}^{\paren{\ell}}}$ and $\mathbf{A}^{\paren{\ell}}$. Hence, the global convergence is not guaranteed \cite{Luo2011IterativeWMMSE,FD_WMMSE}. The computational complexities of WMMSE-MRMC to update $\PiB$, $\PBj$, and $\mathbf{a}\bracket{k}$ at each iteration are, respectively, $\mathcal{O}\paren{\paren{N^\textrm{u}_iD^\textrm{u}_{i}}^3}$, $\mathcal{O}\paren{\paren{M_{\textrm{c}}D^\textrm{d}_{j}}^3}$, and $\mathcal{O}\paren{M_{\textrm{r}}^3}$, primarily because of complexity in solving Sylvester equations. The total per-frame complexity of each iteration for Algorithm \ref{convexalgorithm} to solve \eqref{WMMSE2} is $\mathcal{O}\paren{I^2\paren{N^\textrm{u}_iD^\textrm{u}_{i}}^3}+\mathcal{O}\paren{J^2\paren{M_{\textrm{c}}D^\textrm{d}_{j}}^3}+\mathcal{O}\paren{N^2\rr M^3_{\textrm{r}}}$. GB search on the other hand, is more computationally complex with $\mathcal{O}\paren{I^3\paren{N^\textrm{u}_iD^\textrm{u}_{i}}^3}$ per iteration multiplications for DL \cite{MSE_FD}. The objective function in $\paren{\ref{radarmiop}}$ satisfies the Kurdyka-\L ojasiewicz property. Therefore, the sequence $\braces{\mathbf{a}^{\paren{\ell}}_{m\rr}}$ generated by Algorithm~\ref{Nearness} at the $\ith{\ell}$ step is convergent for all $m\rr$ \cite{arXiv180203889Z}. The BCD-AP MRMC algorithm also converges to the local optimum with a convergence rate of $\mathcal{O}\paren{\sfrac{1}{\mathrm{\ell}_{\textrm{max}}}}$\cite{BCDconvergence}. 
\color{black} In general, initialization methods affect the convergence rate and local optimal values of BCD-AP MRMC. As suggested in \cite{FD_WMMSE,Luo2011IterativeWMMSE}, to reasonably \textit{approach} the global optimum, one can perform random precoder initializations and average over a large number of channel realizations while keeping track of the best result.
\color{black}
\section{Numerical Experiments}
\label{sec:numerical}
We validated our spectral co-design approach through extensive numerical experiments. Throughout this section, we assume the noise variances $\sigma^2_{\textrm{r}}=\sigma^2_{\B}=\sigma^2_{\textrm{d}}=0.001$. We assume unit small scale fading channel gains, namely, the elements of $\mathbf{H}_{\textrm{B},j}$, $\mathbf{H}_{i,\textrm{B}}$, $\mathbf{H}_{i,j}$, $\boldsymbol{\alpha}_{\textrm{Bm},n\rr}$, and $\boldsymbol{\alpha}_{i,n\rr}$ are drawn from $\mathcal{CN}\paren{0,1}$. We model the self-interfering channel $\HBB$ as $\mathcal{CN}\paren{\sqrt{\frac{\sigma^2_{\mathrm{SI}}K_{\B}}{1+K_{\B}}}\widehat{\mathbf{H}}_{\textrm{BB}},\frac{\sigma^2_{\mathrm{SI}}}{1+K_{\B}}\mathbf{I}_{N\cc}\otimes\mathbf{I}_{M\cc}}$, where $\sigma^2_{\mathrm{SI}}$ is the SI attenuation coefficient that characterizes the effectiveness of SI cancellation \cite{MSE_FD}, the Rician factor $K_{\B}=1$, and $\widehat{\mathbf{H}}_{\textrm{BB}}\in\mathbb{C}^{N\cc\times M\cc}$ is an all-one matrix \cite{FD_WMMSE}. Define the signal-to-noise ratios (SNRs) associated with the MIMO radar, DL, and UL as $\mathrm{SNR}_{\textrm{r}}=\sfrac{P_{\textrm{r},m\rr}}{\sigma^2_{\textrm{r}}}$, $\mathrm{SNR}_{\textrm{DL}}=\sfrac{P_{\textrm{B}}}{\sigma^2_{\textrm{d}}}$, and $\mathrm{SNR}_{\textrm{UL}}=\sfrac{P_{\textrm{u},i}}{\sigma^2_{\textrm{\B}}}$ \cite{Luo2011IterativeWMMSE}. The clutter power $\sigma^2_{\textrm{c}}=\sigma^2_{m\rr\textrm{c}n\rr}$ for all $m\rr$ and $n\rr$ and clutter-to-noise ratio (CNR) is $\mathrm{CNR}=\sfrac{\sigma^2_{\textrm{c}}}{\sigma^2_0}$. Then, together with the direct path components, they are received at the IBFD communications Rxs. We model $\mathbf{h}_{m\rr,\textrm{B}}$ and $\mathbf{h}_{m\rr, j}$ as $\mathcal{CN}\paren{\sqrt{\frac{1}{\kappa+1}}\boldsymbol{\mu}_{m\rr,\textrm{B}},\frac{\eta^2_{\mathrm{m\rr},\textrm{B}}}{\kappa+1}\mathbf{I}_{N\cc}}$, and $\mathcal{CN}\paren{\sqrt{\frac{1}{\kappa+1}}\boldsymbol{\mu}_{m\rr,j},\frac{\eta^2_{\mathrm{m\rr},j}}{\kappa+1}\mathbf{I}_{N^{\textrm{d}}_{j}}}$, where $\kappa=1$, $\boldsymbol{\mu}_{m\rr,\B}=0.1\mathbf{1}_{N\cc}$, $\boldsymbol{\mu}_{m\rr,j}=0.05\mathbf{1}_{N^{\textrm{d}}_{j}}$,  $\eta^2_{m\rr,\B}=0.3$, $\eta^2_{m\rr,j}=0.5$.

Unless otherwise stated, we use the following parameter values: number of radar Txs and Rxs: $\mathit{M}\rr=N\rr=4$; number of communications Tx and Rx antennas:  $\mathit{M}\cc=N\cc=4$; $\mathit{I}=\mathit{J}=2$; $N^{\textrm{u}}_{i}=\mathrm{d}^{\textrm{u}}_i=N^{\textrm{d}}_{j}=\mathrm{d}^{\textrm{d}}_j=2$,  for all $\braces{i,j}$; $\mathrm{SNR}_{\textrm{DL}}=\mathrm{SNR}_{\textrm{UL}}=10$ dB; $\sigma^2_{\mathrm{SI}}=0$ dB; \textcolor{black}{$\mathrm{CNR}=20$ dB}; radar PAR $\gamma_{m\rr}=3$ dB; number of communications frames or radar PRIs $\mathit{K}=8$; number of symbols in each frame or range cells in each radar PRI $N=32$; radar CUT index $n_\target=4$; UL (DL) indices of interest $n_{\textrm{rB}}=2$ ($n_{\textrm{r,d}}=3$); QoS of UL (DL): $\mathit{R}_{\textrm{u}}=\log_2(1+\frac{\mathrm{SNR}_{\textrm{UL}}}{M\rr*\mathrm{SNR}_{\textrm{r}}+\mathrm{SNR}_{\textrm{DL}}+(I-1)*\mathrm{SNR}_{\textrm{UL}}})$ bits/s/Hz ($\mathit{R}_{\textrm{d}}=\log_2(1+\frac{\sfrac{\mathrm{SNR}_{\textrm{DL}}}{J}}{M\rr*\mathrm{SNR}_{\textrm{r}}+\mathrm{SNR}_{\textrm{DL}}*\sfrac{(J-1)}{J}+I*\mathrm{SNR}_{\textrm{UL}}})$ bits/s/Hz); 
The normalized Doppler shifts $f_{m\rr\textrm{t}n\rr}T\rr$ and  $f_{\textrm{Bt}n\rr}T\rr$ are uniformly distributed in $\bracket{0.05,0.325}$ for each channel realization \cite{NaghshTSP2017}. The numbers of iterations for the subgradient, weighted minimum mean-squared-error (WMMSE)-MRMC, and BCD-AP MRMC algorithms are $t_{\textrm{u,max}}=t_{\textrm{d,max}}=200$, $\mathrm{\iota}_{\textrm{max}}=1$, $\mathrm{\ell}_{\textrm{max}}=2000$. We use uniform weights  $\alpha^\textrm{u}_{i}= \alpha^{\textrm{d}}_{j}=\alpha^\textrm{r}_{n\rr}=$ $\frac{1}{\paren{\mathit{I}+\mathit{J}+\mathit{N}\rr}}$ for all $\braces{n\rr,i,j}$.

\subsection{Convergence Analysis}
We demonstrate the convergence of the BCD-AP MRMC algorithm with different initialization settings for the FD communications precoders  $\braces{\mathbf{P}^{\paren{0}}}$ with $\mathrm{SNR}_{\textrm{r}}$ equal to $-5$ dB, $0$ dB, $5$ dB, and $10$ dB. \color{black}The first setting, dubbed as the deterministic initialization, initializes $\mathbf{P}^{\paren{0}}_{\textrm{d},j}\bracket{k}$ as the first $D^\textrm{d}_{j}$ columns of the right singular matrix of $\HBj$  and $\mathbf{P}^{\paren{0}}_{\textrm{u},i}\bracket{k}$ as the first $D^\textrm{d}_{j}$ columns of the right singular matrix of $\HiB$. Then, scale the non-zero singular values of $\mathbf{P}^{\paren{0}}_{\textrm{d},j}\bracket{k}$ and $\mathbf{P}^{\paren{0}}_{\textrm{u},i}\bracket{k}$ to be $\sfrac{P_\B}{JD^{\textrm{d}}_{j}}$ and $\sfrac{P_{\textrm{U}}}{D^{\textrm{u}}_{i}}$, respectively, for all $\braces{i,j,k}$ \cite{Luo2011IterativeWMMSE}.  The second method, or the random initialization, generates the singular vectors of $\mathbf{P}^{\paren{0}}_{\textrm{u},i}\bracket{k}$ and $\mathbf{P}^{\paren{0}}_{\textrm{d},j}\bracket{k}$ as two random matrices drawn from $\mathcal{CN}\paren{0,1}$ and normalizes singular values in the same way as the deterministic initialization. The first initialization method offers lower computational complexity than the latter.
\figurename{~\ref{fig:convergence}} shows both initialization approaches achieve convergence as the number of iterations increases for all the simulated radar SNR values. However, the proposed algorithm with random initialization consistently outperforms its counterpart with the deterministic initialization at the cost of computational complexity. The inset plot in \figurename{
~\ref{fig:convergence}} numerically demonstrates the efficiency of the proposed WMMSE-based algorithm; it converges within 80 iterations for both deterministic and random initializations.
\begin{figure}[t]
	\centering
	\includegraphics[width=1.0\columnwidth]{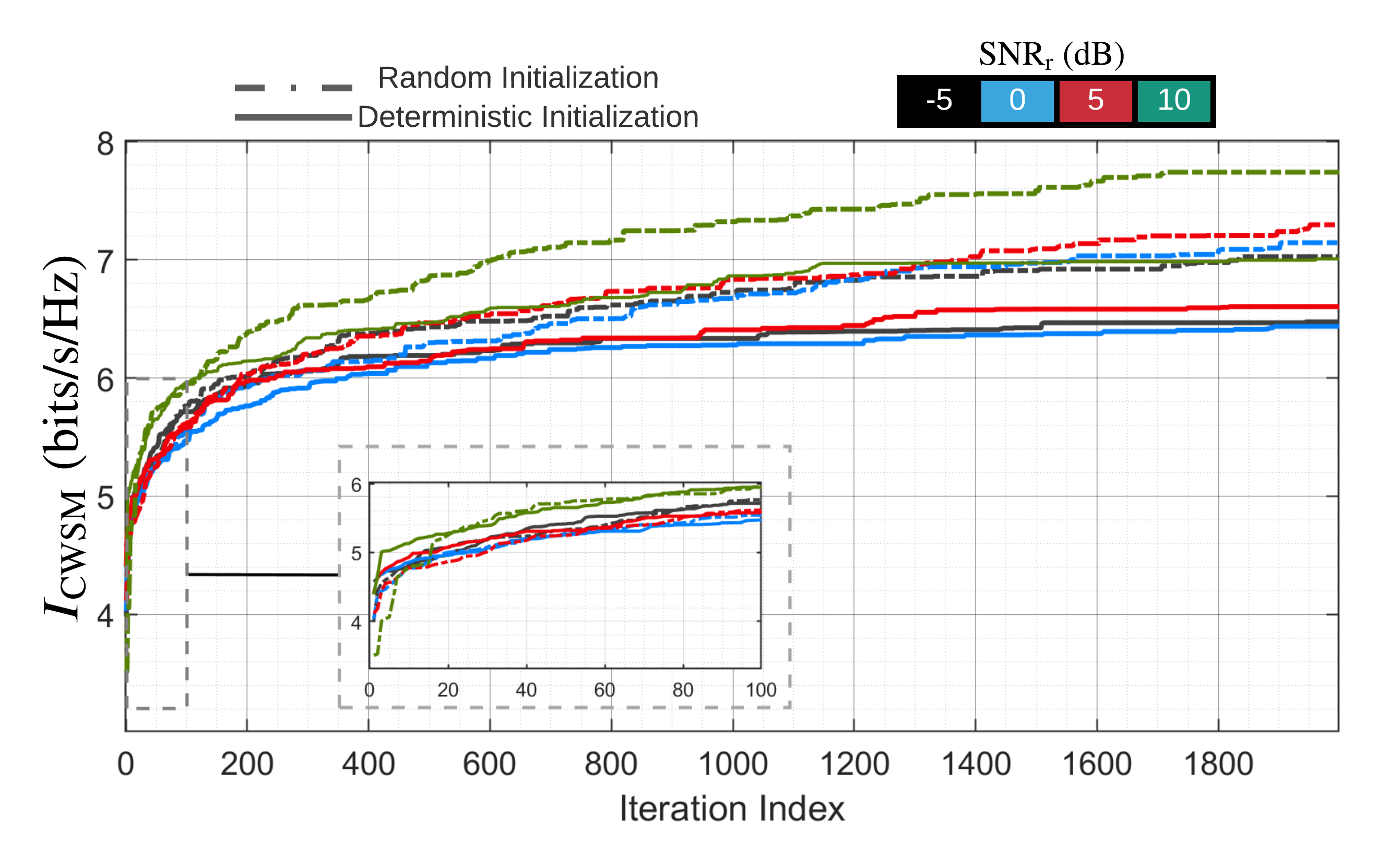}
	\caption{Convergence behaviors of the BCD-AP MRMC algorithm with two initialization methods and multiple $\mathrm{SNR}_\textrm{r}$s. 
	}
	\label{fig:convergence}
\end{figure}

\subsection{Joint Radar-Communications Performance}
\begin{figure}[t]
	\centering
	\includegraphics[width=1.0\columnwidth]{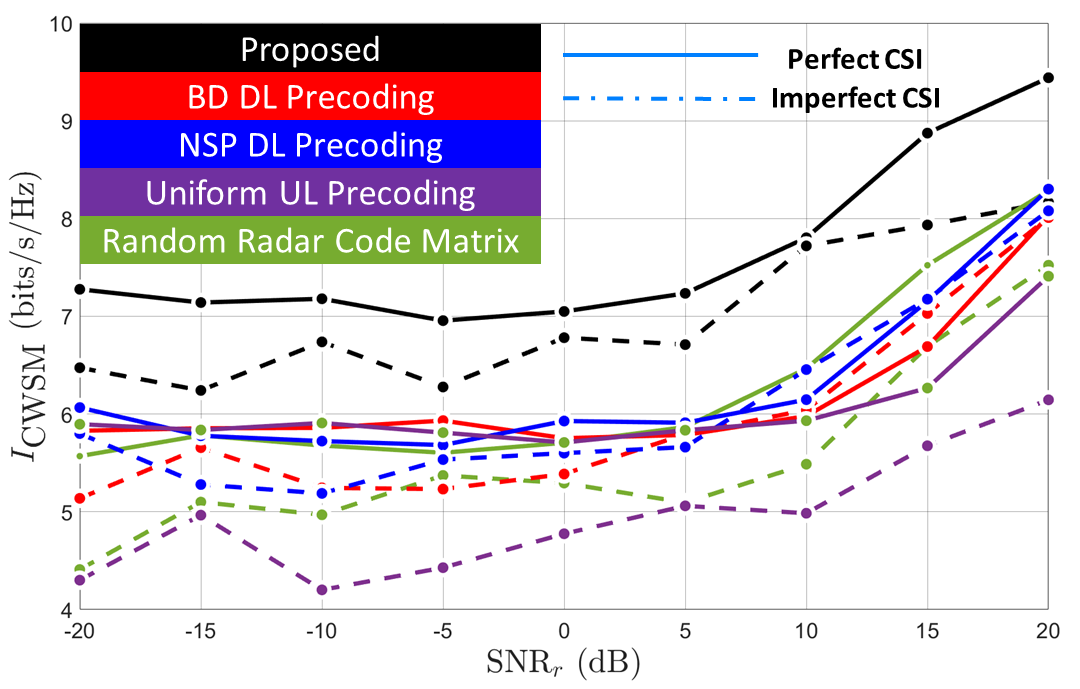}
	\caption{Proposed co-design compared with the conventional communications precoding and radar coding techniques under varying radar SNRs with and without CSI errors.}
	\label{fig:system_vs_SNR}
\end{figure}
\begin{figure*}[t]
    \centering
    \includegraphics[width=\textwidth]{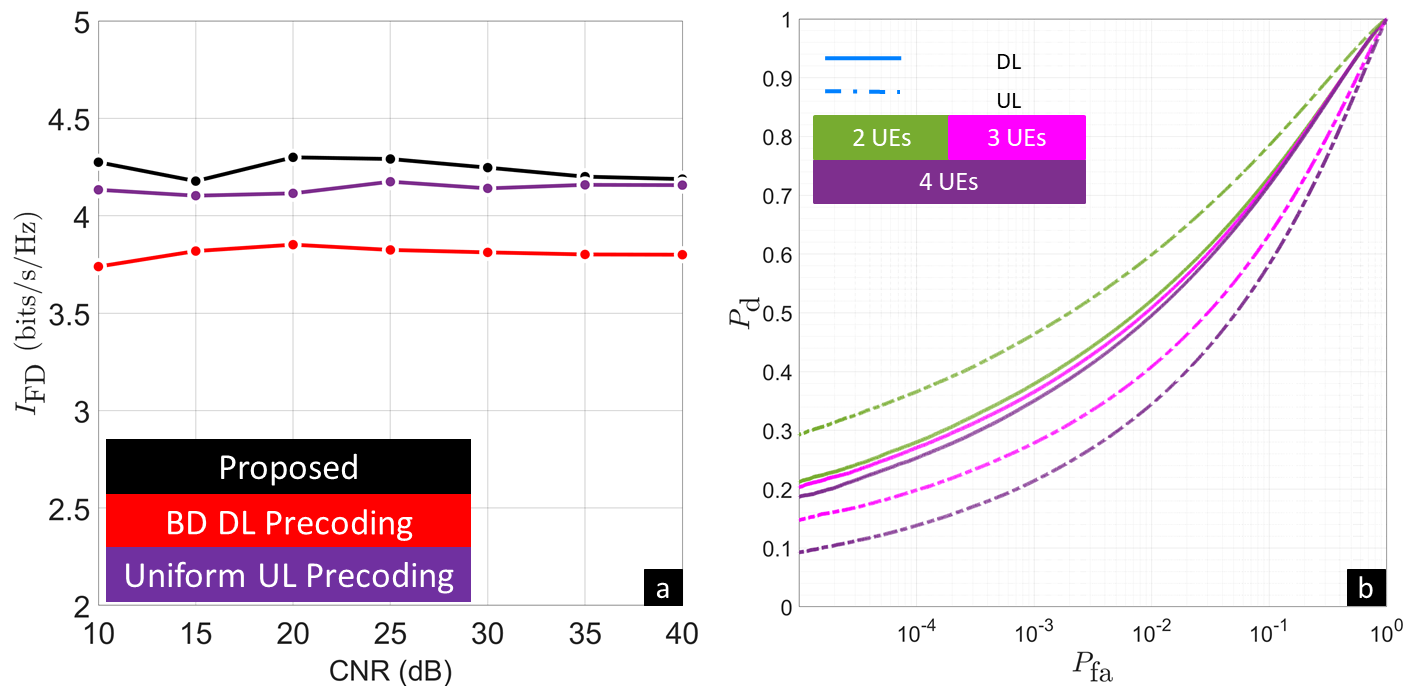}
    \caption{Joint radar and communications analyses: (a) IBFD MU-MIMO performance versus CNRs. (b) ROC curves with varying numbers of UL/DL UEs.}
    \label{fig:combined}
\end{figure*}
\begin{figure*}[t]
    \centering
    \includegraphics[width=\textwidth,trim={85cm 0 0 0},clip]{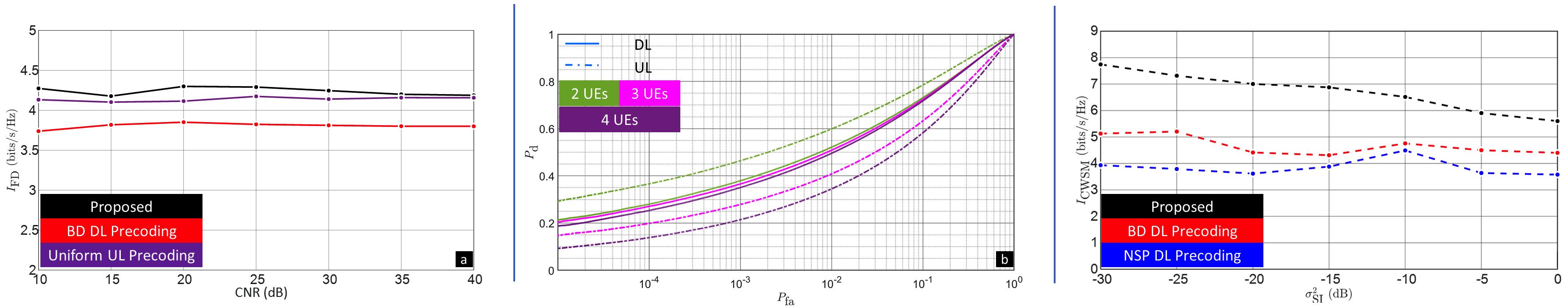}
    \caption{Impact of the FD SI on the proposed radar-communications co-design.}
    \label{fig:fd_si}
\end{figure*}
We evaluated the co-design performance by observing the mutual impact of the statistical MIMO radar and IBFD MU-MIMO communications on each other. \figurename{~\ref{fig:system_vs_SNR}} compares the comprehensive performance of the proposed BCD-AP MRMC algorithm with the existing communications precoding schemes and radar codes when the radar SNR changes\footnote{Note that the proposed UL precoding algorithm is applied to BD and NSP; the proposed DL precoding algorithm is applied to the uniform UL precoding case, and the proposed UL and DL precoding algorithms are applied to the cases of random and uncoded radar codes.} with perfect and imperfect CSI. We set $\eta^2_{\mathrm{CSI}}$ to $0.1$. The proposed co-design algorithm demonstrates a significant advantage over other precoding and radar code matrices for radar SNRs between -20 dB and 20 dB, with and without CIR errors. It is noteworthy that our approach with imperfect CSI yields a higher $I_{\mathrm{CWSM}}$ than even the perfect CSI scenarios of conventional techniques. This is explained by the fact that local CSI for each Tx, i.e., the channel coefficients that are directly connected to this Tx, is needed at each iteration of Algorithm~\ref{convexalgorithm} \cite{MSE_FD}.

We further investigated the impacts of the presence of clutter on the IBFD MU-MIMO communications in \figurename{~\ref{fig:combined}}a and different numbers of communications UEs on the MIMO radar target detection in \figurename{~\ref{fig:combined}}b. We plot the total MI of the IBFD MU-MIMO system $I_{\textrm{FD}}=\sum_{k=0}^{K-1}\bracket{\sum_{i=1}^\mathit{I}\alpha^\textrm{u}_i\mathit{I}^{\textrm{u}}_{i}\bracket{k}+\sum_{j=1}^\mathit{J}\alpha^\textrm{d}_j\mathit{I}^{\textrm{d}}_{j}\bracket{k}}$ versus CNRs in \figurename{~\ref{fig:combined}}a, where $I_{\textrm{FD}}$ remains relatively invariant when CNR increases from $10$ to $40$ dB and the proposed algorithm outperforms other precoding strategies for simulated CNRs. We performed $10^4$ Monte Carlo simulations with $D^{\textrm{u}}_{i}=2$ $(1)$, $D^{\textrm{d}}_{j}=1$ $(2)$ for UL (DL) ROC curves in Fig.~\ref{fig:combined}b. The number of DL UEs is 2 for the dashed curves, and the number of UL UEs is 2 for the solid curves.  which shows that the number of UL UEs is negatively correlated with the MIMO radar detection performance while the cooperation between the DL and the MIMO radar sustains the radar detection performance. \figurename{~\ref{fig:fd_si}} depicts the impact of FD SI attenuation level ranging from $-30$ to $0$ dB\footnote{Note that the NSP performance at $\sigma^2_{\textrm{SI}} = 0$ dB is different than Figure 6 in the companion paper Part I, which employs random initialization. The experiment in this paper used the deterministic initialization.} on the joint radar-communications co-design measured via $I_{\mathrm{CWSM}}$. As expected, the stronger the SI cancellation, the better the MRMC system outputs. Our proposed algorithm outperforms other precoders even at very low SI cancellation or high values of $\sigma^2_{\mathrm{SI}}$. 
\section{Summary}
\label{sec:conclusion}
We jointly designed UL/DL precoders, MIMO radar code matrix, and LRFs for both statistical MIMO radar and IBFD MU-MIMO communications operating in the same frequency bands. We proposed BCD-AP MRMC algorithm that ensures convergence and, as shown in the previous companion paper (Part I) \cite{liu2024codesigningpart1}, delivers performance benefits for both the radar and communications systems. This paper, through extensive experimentation, demonstrated the rapid convergence of the BCD-AP MRMC algorithm using two different initialization schemes. Our co-designed DL and radar systems exhibited robustness against substantial UL interference. Moreover, our optimized precoders and radar codes maintained stable DL and UL data rates as the CNR increased. Furthermore, we validated the robustness of our methods against imperfect CIR estimates and a wide range of SI attenuation levels. In summary, our proposed spectral co-design framework and the BCD-AP MRMC algorithm offer significant performance improvements for both the MIMO radar and IBFD MU-MIMO communications system. The following companion paper (Part III) \cite{liu2024codesigningpart3} considers distributed beamforming and tracking aspects of the same system.

\appendix

\section{Derivation of Gradients}
\label{app:grad}
Denote the complex gradient operator for a scalar real-valued function with a complex-valued matrix argument $f\paren{\mathbf{Z},\mathbf{Z^\ast}}$ as $\nabla_\mathbf{Z}f=\frac{\partial f}{\partial\mathbf{Z}^\ast}$. From the derivative formula $\frac{\partial}{\partial \mathbf{X}^\ast}\trace\paren{\mathbf{B}^\top\mathbf{X}^\dagger\mathbf{CXB}}=\mathbf{C}^\top\mathbf{XBB}^\top+\mathbf{CXBB}^\top $ \cite{MSE_FD}, the gradients of $\Xi_{\text{UL}}$ and $\Xi_{\text{DL}}$ w.r.t. $\PiB$, $\PBj$ and $\mathbf{a}\bracket{k}$ are
\begin{align}\label{GD_constraint}
\nabla_{\PiB}\Xi_{\textrm{UL}}&=2\HiBH\boldsymbol{\xi}_{\textrm{UL}}\mathbf{H}_{i,\B}\PiB-2\alpha^{\textrm{u}}_{i}\HiBH\UiBH\WiB,\nonumber \\
\nabla_{\PBj}\Xi_{\text{DL}}&=2\sum_{g=1}^{\mathit{J}}\mathbf{H}^\dagger_{\B,g}\boldsymbol{\xi}_{\textrm{d},g}\mathbf{H}_{\B,g}\PBj-2\alpha^{\textrm{d}}_j\HBjH\UBjH\WBj,\nonumber\\
\nabla_{\PBj}\Xi_{\text{UL}}&=2\HBBH\boldsymbol{\xi}_{\textrm{UL}}\bracket{k}\HBB\PBj, \nonumber\\
\nabla_{\PiB}\Xi_{\text{DL}}&=2\sum_{g=1}^\mathit{J}\mathbf{H}^\dagger_{i,g}\boldsymbol{\xi}_{\textrm{d},g}\bracket{k}\mathbf{H}_{i,g}\PiB,\nonumber\\
\nabla_{\mathbf{a}\bracket{k}}\Xi_{\text{UL}}&=2\HrBH\boldsymbol{\xi}_{\textrm{UL}}\bracket{k}\HrB\mathbf{a}\bracket{k},
\end{align}
and 
\begin{align}
\nabla_{\mathbf{a}\bracket{k}}\Xi_{\text{DL}}&=2\sum_{g=1}^{\mathit{J}}\HrgH\boldsymbol{\xi}_{\textrm{d},g}\bracket{k}\Hrg\mathbf{a}\bracket{k}, 
\end{align}
respectively, where $\boldsymbol{\xi}_{\textrm{UL}}\bracket{k}=\sum_{q=1}^{\mathit{I}}\alpha^\textrm{u}_q\UqBnH\WqB\UqB$ and  $\boldsymbol{\xi}_{\textrm{d},g}\bracket{k}=\alpha^\textrm{d}_g\mathbf{U}^\dagger_{\textrm{d},g}\bracket{k}$ $\mathbf{W}_{\textrm{d},g}\bracket{k}\mathbf{U}_{\textrm{d},g}\bracket{k}$.
The gradients of $\Xi_{\text{r}}$ w.r.t. $\PiB$, $\PBj$, and
$\mathbf{a}\bracket{k}$ are respectively shown as 
\begin{align}
\nabla_{\PiB}\Xi_{\mathrm{r}}&=2\sum_{n\rr=1}^{N\rr}\sum_{m=1}^{K}\mathrm{Re}\paren{\mathrm{\xi}_{\mathrm{r},n\rr}\paren{k,m}\boldsymbol{\Sigma}^{\paren{m,k}}_{i,n\rr}}\mathbf{P}_{\textrm{u},i}\bracket{m} \mathbf{d}_{\textrm{u},i}\bracket{m,n_\target-n_{\textrm{u}}}\mathbf{d}^\dagger_{\textrm{u},i}\bracket{k,n_{\textrm{t}}-n_{\textrm{u}}}\nonumber
\end{align}
\begin{align}
\nabla_{\PBj}\Xi_{\textrm{r}} &= 2\sum_{n\rr=1}^{N\rr}\sum_{m=1}^{K}\mathrm{Re}\paren{\mathrm{\xi}_{\mathrm{r},n\rr}\paren{k,m}\boldsymbol{\Sigma}^{\paren{m,k}}_{\mathrm{Bt},n\rr}} \sum_{g=1}^{J}\left\lbrace\mathbf{P}_{\textrm{d},g}\bracket{m}\mathbf{d}_{\mathrm{d},g}\bracket{m,0}\mathbf{d}^\dagger_{\textrm{d},j}\bracket{k,0}\right\rbrace\nonumber\\
&+2\sum_{n\rr=1}^{N\rr}\sum_{m=1}^{K}\mathrm{Re}\paren{\mathrm{\xi}_{\mathrm{r},n\rr}\paren{k,m}\boldsymbol{\Sigma}^{\paren{m,k}}_{\mathrm{Bm},n\rr}}\sum_{g=1}^{J}\mathbf{P}_{\textrm{d},g}\bracket{m}\mathbf{d}_{\mathrm{d},g}\bracket{m,n_\target-n_{\textrm{Bm}}}\mathbf{d}^\dagger_{\textrm{d},j}\bracket{k,n_{\textrm{t}}-n_{\textrm{Bm}}}\nonumber\\
&-2\sum_{n\rr=1}^{N\rr}\sum_{m=1}^{K}\mathrm{Re}\paren{\boldsymbol{\Sigma}^{\paren{m,k}}_{\mathrm{Bt},n\rr}}\mathbf{J}^\top_{\mathrm{B}}\mathbf{J}^\top_{\mathrm{h}}\bracket{m}\Wrnr\urk\mathbf{d}^\dagger_{\textrm{d}.j}\bracket{k,0},\nonumber
\end{align}
\begin{align}
\nabla_{\mathbf{a}\bracket{k}}\Xi_{\text{r}}&=
-2\sum_{n\rr=1}^{N\rr}\sum_{m=1}^{K}\mathrm{Re}\paren{\boldsymbol{\Sigma}^{\paren{m,k}}_{\mathrm{rt},n\rr}}\mathbf{J}^\top\rr\mathbf{J}^\top_{\mathrm{h}}\bracket{m}\Wrnr\urk\nonumber\\
&+2\sum_{n\rr=1}^{N\rr}\sum_{m=1}^K\bracket{\mathrm{Re}\paren{\mathrm{\xi}_{\mathrm{r},n\rr}\paren{k,m}\boldsymbol{\Sigma}^{\paren{m,k}}_{\mathrm{rt},n\rr}}+\mathrm{Re}\paren{\mathrm{\xi}_{\mathrm{r},n\rr}\paren{k,m}\boldsymbol{\Sigma}_{\textrm{c},n\rr}}}\mathbf{a}\bracket{m},\nonumber
\end{align}
\normalsize
where  $\mathrm{\xi}_{\mathrm{r},n\rr}\paren{k,m}=\alpha^{\textrm{r}}_{n\rr}\trace\braces{\mathbf{u}^\dagger_{\textrm{r},n\rr}\bracket{k}\mathbf{W}\rnr\mathbf{u}_{\textrm{r},n\rr}\bracket{m}}$.
The derivatives of $f\paren{\mathbf{X},\mathbf{X}^\ast}$ in \eqref{eq: Taylor} w.r.t. $\mathbf{X}^\ast$ are thus approximated as $\frac{\partial f}{\partial \mathbf{X}^\ast}=\frac{\partial }{\partial \mathbf{X}^\ast_{\mathrm{0}}}f\paren{\mathbf{X}_{\mathrm{0}},\mathbf{X}^\ast_{\mathrm{0}}}$ 
The chain rule for a scalar function $g\paren{\mathbf{U\paren{\mathbf{X},\mathbf{X}^\ast}},\mathbf{U}^\ast\paren{\mathbf{X},\mathbf{X}^\ast}}$ where $g$ is dependent on $\mathbf{X}^\ast$ through the matrix $\mathbf{U}$ is \cite{MSE_FD}
\begin{flalign}
\label{eq: chainrule}
\frac{\partial g}{\partial \mathbf{X}^\ast}=\frac{\trace\braces{\paren{\frac{\partial g }{\partial \mathbf{U}}}^\top\partial \mathbf{U}}}{\partial \mathbf{X}^\ast} + \frac{\trace\braces{\paren{\frac{\partial g }{\partial \mathbf{U}^\ast}}^\top\partial \mathbf{U}^\ast}}{\partial \mathbf{X}^\ast}.
\end{flalign}\normalsize
With 
\eqref{eq: chainrule} and 
$\partial\log\left|\mathbf{X}\right|=\trace\braces{\mathbf{X}^{-1}\partial\mathbf{X}}$ \cite{MSE_FD}, the derivatives of $\mathit{R}_{\textrm{u},i}\bracket{k}$ and $\mathit{R}_{\textrm{d},j}\bracket{k}$ w.r.t. $\PiB$, $\PBj$ based on their associated first order Taylor series expansions are 
\begin{align}
\nabla_{\PiB}\mathit{R}_{\textrm{u},i}\bracket{k}&=\HiBH\paren{\mathbf{R}^{\textrm{in}}_{\textrm{u},i}\bracket{k}}^{-1}\HiB\widetilde{\mathbf{P}}_{\textrm{u},i}\bracket{k}\widetilde{\mathbf{E}}^{\star}_{\textrm{u},i}\bracket{k},\nonumber\\ \nabla_{\PBj}\mathit{R}_{\textrm{d},j}\bracket{k}&=\mathbf{H}^\dagger_{\B,j}\paren{\mathbf{R}^{\textrm{in}}_{\textrm{d},j}\bracket{k}}^{-1}\HBj\widetilde{\mathbf{P}}_{\textrm{d},j}\bracket{k}\widetilde{\mathbf{E}}^{\star}_{\textrm{d},j}\bracket{k}, \nonumber\\
\nabla_{\mathbf{P}_{\textrm{u},i}\bracket{k}}\mathit{R}_{\textrm{u},q}\bracket{k}&=-\HiBH\Rinqin\HqB\PqB\mathbf{E}^\star_{\textrm{u},q}\bracket{k}\times\PqBH\HqBH\Rinqin\HiB\widetilde{\mathbf{P}}_{\textrm{u},i}\bracket{k}, ~q\neq i,\nonumber\\
\nabla_{\PBj}\mathit{R}_{\textrm{u},i}\bracket{k}&=-\HBB\Riniin\HiBH\PiB\mathbf{E}^\star_{\textrm{u},i}\bracket{k} \HiBH\PiBH\Riniin\HBB\widetilde{\mathbf{P}}_{\textrm{d},j},\nonumber\\
\nabla_{\PiB}\mathit{R}_{\textrm{d},j}\bracket{k}&=-\mathbf{H}^\dagger_{i,j}\Rinjin\HBj\PBj\mathbf{E}^{\star}_{\textrm{d},j}\bracket{k} \PBjH\HBjH\Rinjin\mathbf{H}_{i,j}\widetilde{\mathbf{P}}_{\textrm{u},i}\bracket{k},\nonumber\\
\nabla_{\PBj} 
\mathit{R}_{\textrm{d},g}\bracket{k}&=-\HBgH\Ringin\HBg\PBg\mathbf{E}^\star_{\textrm{d},g}\bracket{k} \PBgH\HBgH\Ringin\HBg\widetilde{\mathbf{P}}_{\textrm{d},j}\bracket{k},~g\neq j, \nonumber
\end{align}

where $\widetilde{\mathbf{E}}^\star_{\textrm{u},i}\bracket{k}\triangleq\EiBop\paren{\widetilde{\mathbf{P}}_{\textrm{u},i}\bracket{k}}$, $\widetilde{\mathbf{E}}^\star_{\textrm{d},j}\bracket{k}=\EBjop\paren{\widetilde{\mathbf{P}}_{\textrm{d},j}\bracket{k}}$. The derivatives of $\mathit{R}_{\textrm{u},i}\bracket{k}$ and $\mathit{R}_{\textrm{d},j}\bracket{k}$ w.r.t. $\mathbf{a}\bracket{k}$ are
$\nabla_{\mathbf{a}\bracket{k}}\mathit{R}_{\textrm{d},j}\bracket{k}=-\mathbf{H}^\dagger_{\textrm{r},j}\Rinjin$ $\PBj\HBj\mathbf{E}^\star_{\textrm{d},j}\bracket{k}\PBjH\HBjH\Rinjin\mathbf{H}_{\textrm{r},j}\mathbf{a}\bracket{k}$ and $\nabla_{\mathbf{a}\bracket{k}}\mathit{R}_{\textrm{u},i}\bracket{k}=-\HrBH\Riniin\PiB\HiB\mathbf{E}^{\star}_{\textrm{u},q}\bracket{k}\PiBH\HiBH$ $\Riniin\HrB\mathbf{a}\bracket{k}$.
	
	\bibliographystyle{elsarticle-num}
 \bibliography{SP_SI_Part1_v02}

\begin{thebibliography}{10}
\expandafter\ifx\csname url\endcsname\relax
  \def\url#1{\texttt{#1}}\fi
\expandafter\ifx\csname urlprefix\endcsname\relax\def\urlprefix{URL }\fi
\expandafter\ifx\csname href\endcsname\relax
  \def\href#1#2{#2} \def\path#1{#1}\fi

\bibitem{mishra2019toward}
K.~V. Mishra, M.~R. Bhavani~Shankar, V.~Koivunen, B.~Ottersten, S.~A. Vorobyov,
  Toward millimeter wave joint radar communications: {A} signal processing
  perspective, IEEE Signal Processing Magazine 36~(5) (2019) 100--114.

\bibitem{skolnik2008radar}
M.~I. Skolnik, Radar handbook, 3rd Edition, McGraw-Hill, 2008.

\bibitem{Multiuser}
B.~Clerckx, C.~Oestges, {MIMO} wireless networks: Channels, techniques and
  standards for multi-antenna, multi-user and multi-cell systems, Academic
  Press, 2013.

\bibitem{cover2006elements}
T.~M. Cover, J.~A. Thomas, Elements of information theory, 2nd Edition, John
  Wiley \& Sons, 2006.

\bibitem{interferencealignment}
Y.~Cui, V.~Koivunen, X.~Jing, Interference alignment based spectrum sharing for
  {MIMO} radar and communication systems, in: IEEE International Workshop on
  Signal Processing Advances in Wireless Communications, 2018, pp. 1--5.

\bibitem{ayyar2019robust}
A.~Ayyar, K.~V. Mishra, Robust communications-centric coexistence for
  turbo-coded {OFDM} with non-traditional radar interference models, in: IEEE
  Radar Conference, 2019, pp. 1--6.

\bibitem{dokhanchi2019mmwave}
S.~H. Dokhanchi, B.~S. Mysore, K.~V. Mishra, B.~Ottersten, A {mmWave}
  automotive joint radar-communications system, IEEE Transactions on Aerospace
  and Electronic Systems 55~(3) (2019) 1241--1260.

\bibitem{duggal2020doppler}
G.~Duggal, S.~Vishwakarma, K.~V. Mishra, S.~S. Ram, Doppler-resilient
  802.11ad-based ultra-short range automotive joint radar-communications
  system, IEEE Transactions on Aerospace and Electronic Systems 56~(5) (2020)
  4035--4048.

\bibitem{MCMIMO_RadComm}
B.~Li, A.~Petropulu, Joint transmit designs for co-existence of {MIMO} wireless
  communications and sparse sensing radars in clutter, IEEE Transactions on
  Aerospace and Electronic Systems 53~(6) (2017) 2846--2864.

\bibitem{he2019performance}
Q.~{He}, Z.~{Wang}, J.~{Hu}, R.~S. {Blum}, Performance gains from cooperative
  {MIMO} radar and {MIMO} communication systems, {IEEE} Signal Processing
  Letters 26~(1) (2019) 194--198.

\bibitem{haimovich2008mimo}
A.~M. Haimovich, R.~S. Blum, L.~J. Cimini, {MIMO} radar with widely separated
  antennas, {IEEE} Signal Processing Magazine 25~(1) (2008) 116--129.

\bibitem{fisher2006MIMO}
E.~Fishler, A.~Haimovich, R.~S. Blum, L.~J. Cimini, D.~Chizhik, R.~A.
  Valenzuela, Spatial diversity in radars—models and detection performance,
  {IEEE} Transactions on Signal Processing 54~(3) (2006) 823--838.

\bibitem{li2007mimo}
J.~Li, P.~Stoica, {MIMO} radar with colocated antennas, {IEEE} Signal
  Processing Magazine 24~(5) (2007) 106--114.

\bibitem{sun2024widely}
S.~Sun, Y.~Hu, K.~V. Mishra, A.~P. Petropulu, Widely separated {MIMO} radar
  using matrix completion, IEEE Transactions on Radar Systems 2 (2024)
  180--196.

\bibitem{alaee2020information}
M.~Alaee-Kerahroodi, M.~R. Bhavani~Shankar, K.~V. Mishra, B.~Ottersten,
  Information theoretic approach for waveform design in coexisting {MIMO} radar
  and {MIMO} communications, in: IEEE International Conference on Acoustics,
  Speech and Signal Processing, 2020, pp. 1--5.

\bibitem{dokhanchi2020multi}
S.~H. Dokhanchi, M.~R. Bhavani~Shankar, K.~V. Mishra, B.~Ottersten,
  Multi-constraint spectral co-design for colocated {MIMO} radar and {MIMO}
  communications, in: IEEE International Conference on Acoustics, Speech and
  Signal Processing, 2020, pp. 4567--4571.

\bibitem{bao2019precoding}
D.~Bao, G.~Qin, J.~Cai, G.~Liu, A precoding {OFDM MIMO} radar coexisting with a
  communication system, IEEE Transactions on Aerospace and Electronic Systems
  55~(4) (2019) 1864--1877.

\bibitem{khawar2015target}
A.~Khawar, A.~Abdelhadi, C.~Clancy, Target detection performance of spectrum
  sharing {MIMO} radars, {IEEE} Sensors Journal 15~(9) (2015) 4928--4940.

\bibitem{liu2024codesigningpart1}
J.~Liu, K.~V. Mishra, M.~Saquib, Co-designing statistical {MIMO} radar and
  in-band full-duplex multi-user {MIMO} communications -- {Part I}: {S}ignal
  processing, Signal ProcessingUnder review (2024).

\bibitem{Barneto2021FDcommsensing}
C.~B. {Barneto}, S.~D. {Liyanaarachchi}, M.~{Heino}, T.~{Riihonen},
  M.~{Valkama}, Full duplex radio/radar technology: The enabler for advanced
  joint communication and sensing, IEEE Wireless Communications 28~(1) (2021)
  82--88.

\bibitem{MIMOCOMSecrecy}
A.~{Deligiannis}, A.~{Daniyan}, S.~{Lambotharan}, J.~A. {Chambers}, Secrecy
  rate optimizations for {MIMO} communication radar, IEEE Transactions on
  Aerospace and Electronic Systems 54~(5) (2018) 2481--2492.

\bibitem{biswas2018fdqos}
S.~{Biswas}, K.~{Singh}, O.~{Taghizadeh}, T.~{Ratnarajah}, Coexistence of
  {MIMO} radar and {FD MIMO} cellular systems with {QoS} considerations, {IEEE}
  Transactions on Wireless Communications 17~(11) (2018) 7281--7294.

\bibitem{NaghshTSP2017}
M.~M. {Naghsh}, M.~{Modarres-Hashemi}, M.~A. {Kerahroodi}, E.~H.~M. {Alian}, An
  information theoretic approach to robust constrained code design for {MIMO}
  radars, {IEEE} Transactions on Signal Processing 65~(14) (2017) 3647--3661.

\bibitem{liu2022transceiver}
J.~Liu, K.~V. Mishra, M.~Saquib, Transceiver co-design for full-duplex
  integrated sensing and communications, in: IEEE Global Communications
  Conference, 2022, pp. 3821--3826.

\bibitem{Colornoise_waveform}
B.~Tang, J.~Tang, Y.~Peng, {MIMO} radar waveform design in colored noise based
  on information theory, {IEEE} Transactions on Signal Processing 58~(9) (2010)
  4684--4697.

\bibitem{Jammer_game}
X.~Song, P.~Willett, S.~Zhou, P.~B. Luh, The {MIMO} radar and jammer games,
  {IEEE} Transactions on Signal Processing 60~(2) (2012) 687--699.

\bibitem{MSE_FD}
A.~C. Cirik, R.~Wang, Y.~Hua, M.~Latva-aho, Weighted sum-rate maximization for
  full-duplex {MIMO} interference channels, {IEEE} Transactions on
  Communications 63~(3) (2015) 801--815.

\bibitem{FD_WMMSE}
P.~{Aquilina}, A.~C. {Cirik}, T.~{Ratnarajah}, Weighted sum rate maximization
  in full-duplex multi-user multi-cell {MIMO} networks, {IEEE} Transactions on
  Communications 65~(4) (2017) 1590--1608.

\bibitem{Liu2018Gloabalsip}
J.~Liu, M.~Saquib, Transmission design for a joint {MIMO} radar and {MU-MIMO}
  downlink communication system, in: IEEE Global Conference on Signal and
  Information Processing, 2018, pp. 196--200.

\bibitem{Lui2006subg}
{Wei Yu}, R.~{Lui}, Dual methods for nonconvex spectrum optimization of
  multicarrier systems, {IEEE} Transactions on Communications 54~(7) (2006)
  1310--1322.

\bibitem{BCDconvergence}
A.~Beck, L.~Tetruashvili, On the convergence of block coordinate descent type
  methods, {SIAM} J. Optim. 23 (2013) 2037--2060.

\bibitem{Liu2017asilomar}
J.~Liu, M.~Saquib, Joint transmit-receive beamspace design for colocated {MIMO}
  radar in the presence of deliberate jammers, in: Asilomar Conf. Signals Syst.
  Comput., 2017, pp. 1152--1156.

\bibitem{ADMMBCD}
C.~{Chen}, M.~{Li}, X.~{Liu}, Y.~{Ye}, Extended {ADMM} and {BCD} for
  nonseparable convex minimization models with quadratic coupling terms:
  convergence analysis and insights, Math. Program. 173 (2019) 37--77.

\bibitem{zhang2017convergent}
Z.~Zhang, M.~Brand, Convergent block coordinate descent for training {Tikhonov}
  regularized deep neural networks, in: Advances Neural Inf. Process. Syst.,
  2017, pp. 1721--1730.

\bibitem{Lops2019serveillance}
E.~{Grossi}, M.~{Lops}, L.~{Venturino}, Joint design of surveillance radar and
  {MIMO} communication in cluttered environments, {IEEE} Transactions on Signal
  Processing 68 (2020) 1544--1557.

\bibitem{arXiv180203889Z}
Z.~Zhu, X.~Li, Convergence analysis of alternating projection method for
  nonconvex sets, arXiv preprint arXiv:1802.03889v2 (2019).

\bibitem{nearestvector}
J.~A. {Tropp}, I.~S. {Dhillon}, R.~W. {Heath}, T.~{Strohmer}, Designing
  structured tight frames via an alternating projection method, {IEEE} Trans.
  Inf. Theory 51~(1) (2005) 88--209.

\bibitem{Luo2011IterativeWMMSE}
Q.~{Shi}, M.~{Razaviyayn}, Z.~{Luo}, C.~{He}, An iteratively weighted {MMSE}
  approach to distributed sum-utility maximization for a {MIMO} interfering
  broadcast channel, {IEEE} Transactions on Signal Processing 59~(9) (2011)
  4331--4340.

\bibitem{Boydsubgradientnote}
S.~Boyd, J.~Park,
  \href{https://web.stanford.edu/class/ee364b/lectures/subgrad_method_notes.pdf}{Subgradient
  methods} (May 2014).
\newline\urlprefix\url{https://web.stanford.edu/class/ee364b/lectures/subgrad_method_notes.pdf}

\bibitem{liu2024codesigningpart3}
S.~Nayemuzzaman, J.~Liu, K.~V. Mishra, M.~Saquib, Co-designing statistical
  {MIMO} radar and in-band full-duplex multi-user {MIMO} communications --
  {Part III}: {M}ulti-target tracking, Signal ProcessingUnder review (2024).

\end{thebibliography}
	
\end{document}